\documentclass{llncs}

\usepackage{amsmath}
\usepackage{amssymb}
\usepackage{xspace}
\usepackage{stmaryrd}
\usepackage[usenames]{color}
\usepackage{gastex}
\usepackage{cite}
\usepackage{caption}

\usepackage[usenames]{color}
\usepackage{gastex}
\usepackage{url}
\usepackage{relsize}
\usepackage{graphicx}
\usepackage{latexsym}
\usepackage{boxedminipage}

\usepackage{multicol}
\usepackage{wrapfig}

\usepackage[nounderscore]{syntax}

\newcommand{\ignore}[1]{}
\newcommand{\tuple}[1]{\langle #1 \rangle}

\newcommand{\lub}{\sqcup}
\newcommand{\glb}{\sqcap}
\newcommand{\true}{\textit{true}}
\newcommand{\false}{\textit{false}}

\newcommand{\Strand}{\textsc{Strand}\xspace}

\newcommand{\A}{\mathcal{A}}
\newcommand{\B}{\mathcal{B}}
\newcommand{\AEL}{\A_\text{el}}
\newcommand{\yblank}{-}
\newcommand{\blank}{\underline{b}}
\newcommand{\Lval}{L_{\text{v}}}

\newcommand{\next}{\texttt{next}}

\newcommand{\Loc}{\mathit{Loc}}
\newcommand{\PV}{\mathit{PV}}
\newcommand{\dom}{\mathbb{D}}
\newcommand{\DV}{\mathit{DV}}
\newcommand{\data}{\texttt{data}}

\newcommand{\nil}{{\texttt{nil}}}
\newcommand{\pc}{\mathit{pc}}
\newcommand{\heap}{H}

\newcommand{\pval}{\mathit{pval}}
\newcommand{\dval}{\mathit{dval}}
\newcommand{\PC}{\mathit{PC}}

\newcommand{\dirty}{\mathit{dirty}}

\newcommand{\run}{\rho}
\newcommand{\QSDA}{\mbox{\sc QSDA}}
\newcommand{\EQSDA}{\mbox{\sc EQSDA}}

\newcommand{\F}{\mathcal{F}}
\newcommand{\HC}{\mathcal{H}} 

\begin{document}

\title{Quantified Data Automata on Skinny Trees:\\
an Abstract Domain for Lists}

\author{
Pranav Garg\inst{1}, P. Madhusudan\inst{1}, Gennaro Parlato\inst{2}}

\institute{University of Illinois at Urbana-Champaign, USA\\
\and University of Southampton, UK}

\maketitle

\begin{abstract}
We propose a new approach to heap analysis through an abstract domain
of automata, called \emph{automatic shapes}. The abstract domain uses a particular kind of automata, called \emph{quantified data automata on skinny trees} (\QSDA s), that allows to define universally quantified properties of singly-linked lists.
To ensure convergence of the abstract fixed-point computation, we introduce a sub-class of \QSDA s called elastic \QSDA s, which also form an abstract domain.
We evaluate our approach on several list manipulating programs and we show that the proposed domain is powerful enough to prove a large class of these programs correct.
\end{abstract}

\section{Introduction}

The abstract analysis of heap structures is an important problem in program verification as dynamically evolving heap is ubiquitous in modern programming, either in terms of low level pointer manipulation or in object-oriented programming.
Abstract analysis of the heap is hard because abstractions need to represent the heap that is of unbounded size, and must capture both the \emph{structure} of the heap as well as the unbounded \emph{data} stored in the heap. While several data-domains have been investigated for data stored in static variables, the analysis of unbounded structure and unbounded data that a heap contains has been less satisfactory. The primary abstraction that has been investigated is the rich work on \emph{shape analysis}~\cite{shapeanalysis}. However, unlike abstractions for data-domains (like intervals, octagons, polyhedra, etc.), shape analysis requires carefully chosen \emph{instrumentation} predicates to be given by the user, and often are particular to the program that is being verified. Shape analysis techniques typically \emph{merge} all nodes that satisfy the same unary predicate, achieving finiteness of the abstract domain, and interpret the other predicates using a 3-valued (must, must not, may) abstraction. Moreover, these instrumentation predicates often require to be encoded in particular ways
(for example, capturing binary predicates as particular kinds of unary predicates) so as to not lose precision.

For instance, consider a sorting algorithm that has an invariant of the form:\\
$~~~~~~~~~~~~~~~~~~~~~~~\forall x, y. \left(~(x \rightarrow_\textit{\next}^* y \wedge y \rightarrow_\textit{\next}^* i) \Rightarrow d(x) \leq d(y)~\right)$\\
which says that the sub-list before pointer $i$ is sorted.
In order to achieve a shape-analysis algorithm that discovers this invariant
 (i.e., captures this invariant precisely during the analysis),
we typically need instrumentation predicates such as $p(z) = z \rightarrow_\textit{\next}^* i$,
$s(x) = \forall y. ((x \rightarrow_\textit{\next}^* y \wedge  y \rightarrow_\textit{\next}^* i) \Rightarrow d(x) \leq d(y))$, etc.
The predicate $s(x)$ says that the element that is at $x$ is less than
or equal to the data stored in every cell between $x$ and $i$.
These instrumentation predicates are clearly too dependent on the precise program and property
 being verified.

In this paper, we investigate an abstract domain for heaps that works \emph{without
user-defined instrumentation predicates} (except we require that the user fix an
abstract domain for data, like octagons, for comparing data elements).

We propose a radically new approach to heap analysis through an abstract domain
of automata, called \emph{automatic shapes} (automatic because we use automata).
The abstract domain uses a particular kind of automata, called \emph{quantified data automata}, that define, logically, universally quantified properties of heap structures. In this paper, we restrict our attention to heap structures that
have only \emph{one pointer field}; our analysis is hence one that can be used to analyze
properties of heaps containing lists, with possible aliasing (merging) of them,
especially at intermediate stages in the program. One-pointer heaps can be viewed
as \emph{skinny trees} (trees where the number of branching nodes is bounded).

Automata, in general, are classical ways to capture an infinite set of objects using
finite means. A class of (regular) skinny trees can hence be represented using tree
automata, capturing the structure of the heap. While similar ideas have been
explored before in the literature~\cite{forest-automata}, our main aim is to also represent properties of
the \emph{data} stored in the heap, building automata that can express
universally quantified properties on lists, in particular those
of the form\\\newline
$~~~~~~~~~~~~~~~~~~~~~\bigwedge_i \forall \overline{x}. \left(\textit{Guard}_i(\overline{p},\overline{x}) \Rightarrow \textit{Data}_i(d(\overline{p}), d(\overline{x}))\right)$ \\\newline
where $\overline{p}$ is the set of static pointer variables in the program.
The $\textit{Guard}_i$ formulas express structural constraints on the quantified variables
and the pointer variables, while the $\textit{Data}_i$ formulas express properties about the
data stored at the nodes pointed to by these pointers. In this paper, we investigate
an abstract domain that can infer such quantified properties, parameterized by an abstract
numerical domain ${\cal F}_d$ for the data formulas and by the number of quantified variables $\overline{x}$.

The salient aspect of the automatic shapes that we build is that (a) there is no requirement
from the user to define instrumentation predicates for the structural $\textit{Guard}$
formulas; (b) since the abstraction will not be done by merging unary predicates and since
the automata can define how data stored at \emph{multiple} locations on the heap are related,
there is no need for the user to define carefully crafted unary predicates that relate structure and data
(e.g., like the unary predicate $s(x)$ defined above that says that the location $x$ is sorted with respect
to all successive locations that come after $x$ but before $i$). Despite this lack of help
from the user, we show how our abstract domain can infer properties of a large number
of list-manipulating programs adequately to prove interesting quantified properties.

The crux of our approach is to use a class of automata, called quantified data automata on skinny trees (\QSDA),
to express a class of single-pointer heap structures and the data contained in them. \QSDA s read skinny
trees with data along with \emph{all} possible valuations of the quantified variables, and for each of them
check whether the data stored in these locations (and the locations pointed to by pointer variables in the
program) relate in particular ways defined by the abstract data-domain ${\cal F}_d$. We show that the
class of \QSDA s (over a data-domain ${\cal F}_d$ and a set of variable $\overline{x}$) forms an abstract
data domain lattice. Along with the natural concretization and abstraction relations, this class forms
a Galois connection with respect to the class of concrete single-pointer heap data structures.

We further show, for a simple heap-manipulating programming language, that we can define an abstract
post operator over the abstract domain of \QSDA s. This abstract post preserves the structural aspects of the
heap \emph{precisely} (as \QSDA s can have an arbitrary number of states to capture the evolution of the
program) and that it soundly abstracts the quantified data properties. The abstract post is nontrivial
to define and show it effective as it requires automata-theoretic operations that need to simultaneously
preserve structure as well as data properties; this forms the hardest technical aspect of our paper.
We thus obtain an effective abstract interpretation using the domain of \QSDA s.

Traditionally, in order to handle loops and reach termination, abstract domains require some form of widening.
Our notion of widening is \emph{directed by decidability considerations}. Assume that the programmer computes
a \QSDA\ as an invariant for the program at a particular point, where there is an assertion expressed as
a quantified property $p$ over lists (such as  ``the list pointed to by \emph{head} is sorted'').
In order to verify that the abstraction proves the assertion, we will have to check if the language of lists
accepted by the \QSDA\ is contained in the language of lists that satisfy the property $p$. However, this is in
general \emph{undecidable}. Our aim is to \emph{overapproximate} the \QSDA\ into a larger language accepted
by a particular kind of data automata, called  {\em elastic} \QSDA\ (\EQSDA) for which this inclusion problem is
decidable (for an appropriately chosen language for expressing assertions).

This \emph{elastification} will in fact serve as the basis for widening as well, as there are only a \emph{finite} number of
elastic \QSDA s that express structural properties, discounting the data-formulas. Consequently, we
can combine the elastification procedure (which overapproximates a \QSDA\ into an elastic \QSDA) and
widening over the numerical domain for the data in order to obtain widening procedures that can be
used to accelerate the computation for loops. In fact, the domain of \EQSDA s can be seen as an abstract
domain, and there is a natural abstract interpretation between \QSDA s and \EQSDA s, where the \EQSDA s permit
widening procedures. We show a unique elastification theorem that shows
that for any \QSDA, there is a unique elastic \QSDA\ that over-approximates it. This elastification is in fact
the abstract map $\alpha$ that connects \QSDA s with \EQSDA s (the $\gamma$ map being identity, as \EQSDA s are also
\QSDA s).

We also show that \EQSDA\ properties over lists can be translated to a decidable fragment of the logic {\sc Strand}~\cite{popl11} over lists, and hence
inclusion checking an elastic \QSDA\ with respect to any assertion that is also written using the decidable
sublogic of {\sc Strand} over lists is decidable. The notion of \QSDA s and elasticity are extensions
of recent work in~\cite{CAVQDA}, where such notions were developed for \emph{words}
(as opposed to trees) and where the automata were used for \emph{learning} invariants from examples and counter-examples.

We implement our abstract domain and transformers and show, using a suite of list-manipulating programs,
that our abstract interpretation is able to prove the naturally required (universally-quantified) properties of these programs.
While
several earlier approaches (such as shape analysis) can tackle the correctness of these programs as well, our abstract
analysis is able to do this \emph{without} requiring program-specific help from the user (for example, in terms of
instrumentation predicates in shape analysis, and in terms of guard patterns in the work by Bouajjani et al~\cite{celia}).

%
\paragraph{\bf Related Work.}

Shape analysis~\cite{shapeanalysis} is the one of the most well-known technique for synthesizing invariants about dynamically evolving heaps. However, shape analysis requires user-provided instrumentation predicates which are often too particular to the program being verified. Hence coming up with these instrumentation predicates is not an easy task.
In recent work~\cite{sas10,chang-rival,celia,gulwani08}, several abstract domains have been explored which combine the shape and the data constraints.
Though some of these domains~\cite{sas10,chang-rival} can handle heap structures more complex than singly-linked lists, all these domains require the user to provide a set of data predicates~\cite{gulwani08} or a set of structural guard patterns~\cite{celia} or predicates over both the structure and the data constraints~\cite{sas10,chang-rival}.
In contrast, the only assistance our technique requires from the user is specifying the number of universally quantified variables.

For singly-linked lists,~\cite{rama} introduces a family of abstractions based on a set of instrumentation predicates which track uninterrupted list segments. However these abstractions only handle structural properties and not the more-complex quantified data properties.
Several separation logic based shape analysis techniques have also been developed over the years~\cite{distefano,guo,berdine,slayer}. But they too mostly handle only the shape properties (structure) of the heap.


Our automaton model for representing quantified invariants over lists is inspired by the decidable fragment of \Strand~\cite{popl11} and can track invariants with guard constraints of the form $y \leq t$ or $t \leq y$ for a universal variable $y$ and some term $t$. These structural constraints on the guard are very similar to array partitions in~\cite{gopan,halbwachs-pldi08,cousot-logozzo}. However, our automata model is more general. For instance, none of these related works can handle sortedness of arrays which requires quantification over more than one variable.

Techniques based on \emph{Craig's interpolation} have recently emerged as an orthrogonal way for synthesizing quantified invariants over arrays and lists~\cite{mcmillan06,mcmillan08,natasha,podelski}.
These methods use different heuristics like term abstraction~\cite{podelski} or introduction of existential ghost variables~\cite{natasha} or finding interpolants over a restricted language~\cite{mcmillan06,mcmillan08} to ensure the convergence of the interpolant from a small number of spurious counter-examples.
The shape analysis proposed in~\cite{podelski07} is also counter-example driven.~\cite{podelski07} requires certain quantified predicates to be provided by the user. Given these predicates, it uses a CEGAR-loop for incrementally improving the precision of the abstract transformer and also discovering new predicates on the heap objects that are part of the invariant.

Automata based abstract interpretation has been explored in the past~\cite{forest-automata} for inferring shape properties about the heap. However, in this paper we are interested in strictly-richer universally quantified properties on the data stored in the heap.~\cite{streaming-transducer} introduces a streaming transducer model for algorithmic verification of single-pass list-processing programs. However the transducer model severely constrains the class of programs it can handle; for example,~\cite{streaming-transducer} disallows repeated or nested list traversals which are required in sorting routines, etc.

In this paper we introduce a class of automata called quantified skinny-tree data automata (QSDA) to capture universally quantified properties over skinny-trees. The QSDA model is an extension of recent work in~\cite{CAVQDA} where a similar automata model was introduced for words (as opposed to trees). Also,
the automata model in~\cite{CAVQDA} was parameterized by a \emph{finite} set of data formulas and was used for \emph{learning} invariants from examples and counter-examples. In contrast, we extend the automata in~\cite{CAVQDA} to be instantiated with a (possibly-infinite) abstract domain over data formulas and develop a theory of abstract interpretation over QSDAs.

\section{Programs Manipulating Heap and Data}

We consider sequential programs manipulating acyclic singly-linked data structures. A {\em heap structure} is composed of locations (also called nodes). Each location is endowed  with a {\em pointer field} $\next$ that points to another location or it is undefined,  and a {\em data field} called $\data$ that takes values from a potentially infinite domain $\dom$ (i.e. the set of integers). For simplicity we assume a special location, called $\dirty$, that models an un-allocated memory space. We assume that the $\next$ pointer field of $\dirty$ is undefined. Besides the heap structure, a program also has a finite number of {\em pointer variables} each pointing to a location in the heap structure, and a finite number of {\em data variables} over $\dom$. In our programming language  we do not have procedure calls, and we handle non-recursive procedures calls by inlining the code at call points. In the rest of the section we formally define the syntax and semantics of these programs.

\begin{wrapfigure}[9]{r}[20pt]{83mm}
\vspace{-1.2cm}
  \centering
\scriptsize
  \begin{grammar}
    <prgm> ::=  pointer $p_1,\ldots,p_k$; data $d_1,\ldots,d_\ell$;  <pc\_stmt>$^+$
\vspace{-0.2cm}

    <pc\_stmt> ::= $pc:$ <stmt>;
\vspace{-0.2cm}

      <stmt> ::= <ctrl\_stmt> | <heap\_stmt>
\vspace{-0.2cm}


     <ctrl\_stmt> ::= $d_i:=$<data\_expr> $\mid$ {\tt skip}
 $\mid$ {\tt assume}(<pred>) 
\alt {\tt if} <pred> {\tt then} <pc\_stmt>$^+$ {\tt else} <pc\_stmt>$^+$ {\tt fi}
\alt {\tt while} <pred> {\tt do} <pc\_stmt>$^+$ {\tt od}

\vspace{-0.2cm}

     <heap\_stmt> ::= \mbox{$\texttt{new } p_i$} $\mid$ \mbox{$p_i := \nil$}  $\mid$  \mbox{$p_i := p_j$} \alt \mbox{$p_i := p_j\rightarrow\next$} $\mid$ \mbox{$p_i\rightarrow\next := \nil$} $\mid$ \mbox{$p_i\rightarrow\next := p_j$} \alt $p_i\rightarrow\data :=<data\_expr>$

\end{grammar}
\vspace{-0.5cm}

\caption{Simple programming language.}
\label{grammar}
\end{wrapfigure}

\ignore{

\begin{wrapfigure}[h]{r}[20pt]{83mm}
  \begin{center}
\vspace*{-35pt}

{\bf \scriptsize
  \begin{grammar}
    <prgm> ::=  pointer $p_1,\ldots,p_k$; data $d_1,\ldots,d_\ell$;  <pc\_stmt>$^+$

\vspace*{-5pt}
    <pc\_stmt> ::= $pc:$ <stmt>;

  \vspace*{-5pt}
      <stmt> ::= <ctrl\_stmt> | <heap\_stmt>

  \vspace*{-5pt}

     <ctrl\_stmt> ::= $d_i:=$<data\_expr> $\mid$ {\tt skip}
    $\mid$ {\tt assume}(<pred>) 
\alt {\tt if} <pred> {\tt then} <pc\_stmt>$^+$ {\tt else} <pc\_stmt>$^+$ {\tt fi}
\alt {\tt while} <pred> {\tt do} <pc\_stmt>$^+$ {\tt od}

  \vspace*{-5pt}

     <heap\_stmt> ::= \mbox{$\texttt{new } p_i$} $\mid$ \mbox{$p_i := \nil$}  $\mid$  \mbox{$p_i := p_j$} \alt \mbox{$p_i := p_j\rightarrow\next$} $\mid$ \mbox{$p_i\rightarrow\next := \nil$} $\mid$ \mbox{$p_i\rightarrow\next := p_j$} \alt $p_i\rightarrow\data :=<data\_expr>$



\end{grammar}
}

  \end{center}

\end{wrapfigure}
}

\paragraph{Syntax.}
The syntax of programs is defined by the BNF grammar of Figure~\ref{grammar}.
A program starts with the declaration of pointer variables among which one called $\nil$, followed by a declaration of data variables. Data variables range over a potentially infinite data domain $\dom$. We assume a language of data expressions built from data variables and terms of the form
$p_i\rightarrow \data$ (with $p_i\not=\nil$) using operations over $\dom$.
Predicates in our language are either data predicates built from predicates over $\dom$ or structural predicates concerning the heap built from atoms of the form $p_i == p_j$, $p_i\rightarrow\next == p_j$, and  $p_i\rightarrow^*\next\ == p_j$, for some $i,j\in[1,k]$. Thereafter, there is a non-empty list of labelled statements of the form
$\mathit{pc}\!:\! \langle \mathit{stmt}\rangle$ where $\mathit{pc}$ is the {\em program counter} and $\langle\mathit{stmt}\rangle$ defines a language of either C-like statements or statements which modify the heap.  We do not have an explicit statement to {\em free} locations of the heap: when a location is no longer reachable from any location pointed by a pointer variable we assume that it automatically disappears from the memory. For a program $P$, we denote with $\PC$ the set of all program counters of $P$ statements.
Figure~\ref{fig:prog}(a) shows the code for program \emph{sorted list-insert} which is a running example in the paper. The program inserts a \emph{key} into the sorted list pointed to by variable \emph{head}.

\paragraph{Semantics.} A {\em configuration} $C$ of a program $P$ with set of pointer variables $\PV$ and data variables $\DV$ is a tuple $\tuple{\pc, \heap, \pval, \dval}$ where
\begin{itemize}
\item $\pc\in\PC$ is the program counter of the next statement to be executed;

\item $H$ is a {\em heap configuration} represented by a tuple $(\Loc, \next, \data)$ where (1) $\Loc$ is a finite set of heap locations containing a special element called $\dirty$, (2) $\next: \Loc \mapsto \Loc$ is a partial map defining an edge relation among locations such that the graph $(\Loc,\next)$ is acyclic, and (3) $\data:\Loc \mapsto \dom$ is a  map that associates each location of $\Loc$ with a data value in $\dom$;
\item $\pval:\PV\mapsto \Loc$ associates each pointer variable of $P$ with a location in $H$. If $\pval(p)=v$ we say that node $v$ is {\em pointed} by variable $p$. Furthermore, each node in $\Loc$ is reachable from a node pointed by a variable in $\PV$. There is no outgoing ($\next$) edge from location $\dirty$ and there is a $\next$ edge from the location pointed by $nil$ to $\dirty$;
\item $\dval:\DV \mapsto\dom$ is a valuation map for the data variables.
\end{itemize}

Figure~\ref{fig:prog}(b) graphically shows a progam configuration which is reachable at program counter $8$ of the program in Figure~\ref{fig:prog}(a) (as explained later we encode the data variable \emph{key} as a pointer variable in the heap configuration).
The {\em transition relation} of a program $P$, denoted $\xrightarrow{\mathit{stmt}}_P$ for each statement $\mathit{stmt}$ of $P$, is defined as usual.  The control-flow statements update the program counter, possibly depending on a predicate (condition).
The assignment statements update the variable valuation or the heap structure other than moving to the next program counter. A formal semantics of programs can be found in Appendix~\ref{semantics}.
Let us define the concrete transformer $F^\natural = \lambda \mathcal{C}. \{ \mathcal{C'} \mid \mathcal{C} \xrightarrow{\mathit{stmt}}_P \mathcal{C'} \}$. The concrete semantics of a program is given as the least fixed point of a set of equations of the form $\psi = F^\natural (\psi)$. 

To simplify the presentation of the paper, we assume that our programs do not have data variables. This restriction, indeed, does not reduce their expressiveness: we can always transform a program $P$ into an {\em equivalent} program $P'$ by translating each data variable $d$ into a pointer variable that will now point to a fresh node in the heap structure, in  which the value $d$ is now encoded by $d\rightarrow\data$. The node pointed by $d$ is not pointed by any other pointer, further, $d\rightarrow\next$ points to $\dirty$. Obviously, wherever $d$ is used in $P$ will now be replaced by $d\rightarrow\data$ in $P'$. 

\section{Quantified Skinny-Tree Data Automata} \label{sec:qda}
In this section we define {\em quantified skinny-tree data automata} (\QSDA s, for short), an accepting  mechanism  of program configurations (represented as special labelled trees) on which we can express properties of the form \\
$\bigwedge_i \forall y_1,\ldots,y_\ell.$ $Guard_i$ $\Rightarrow$ $Data_i$,
where variables $y_i$ range over the set of locations of the heap, $Guard_i$ represent quantifier-free structural constraints among the pointer variables and the universally quantified variables $y_i$, and $Data_i$ (called {\em data formulas}) are quantifier-free formulas that refer to the data stored at the locations pointed either by the universal variables $y_i$ or the pointer variables, and compare them using operators over the data domain.  In the rest of this section, we first define {\em heap skinny-trees} which are a suitable labelled tree encodings for program configurations; we then define {\em valuation trees} which are heap skinny-trees by adding to the labels an instantiation of the universal variables. {\em Quantified skinny-tree data automata} is a mechanism designed to recognize valuation trees. The {\em language} of a \QSDA\ is the set of all heap skinny-trees such that all valuation trees deriving from them are accepted by the \QSDA. Intuitively, the heap skinny-trees in the language defined by the \QSDA\ are all the program configurations that verify the formula $\bigwedge_i \forall y_1,\ldots,y_\ell. Guard_i \Rightarrow Data_i$.

Let $T$ be a tree. A node $u$ of $T$ is {\em branching} whenever $u$ has more than one child. For a given natural number $k$, $T$  is  $k$-{\em skinny} if it contains at most $k$ branching nodes.

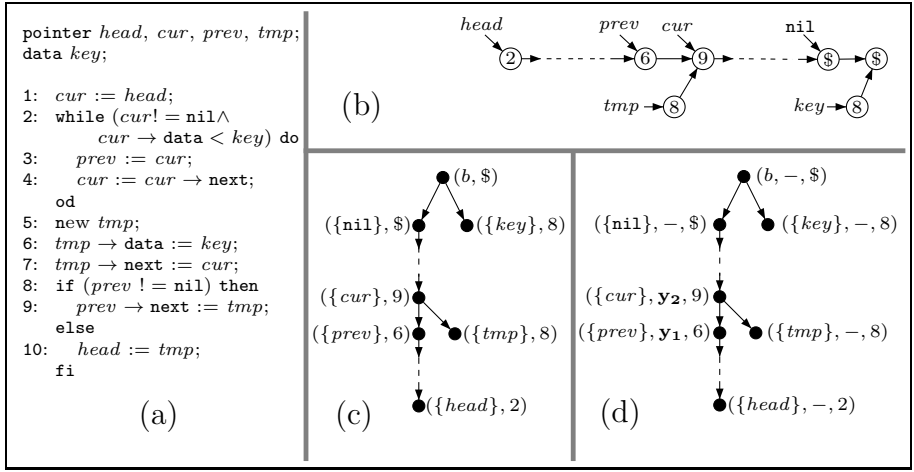
\begin{figure}[tb]
\framebox{
\scriptsize
\begin{minipage}{330pt}

\begin{multicols}{2}

{\tt pointer} $\mathit{head}$, $\mathit{cur}$, $\mathit{prev}$, $\mathit{tmp}$;\\
{\tt data} $\mathit{key}$;\\
\\
{\tt 1}:\hspace*{0.11truecm}  $\mathit{cur}$ := $\mathit{head}$;\\
{\tt 2}:\hspace*{0.13truecm}  {\tt while} ($cur != \nil \wedge$\\
\hspace*{1.0truecm}$cur\rightarrow\data < key$) {\tt do}\\
{\tt 3}:\hspace*{0.4truecm} $\mathit{prev}$ := $\mathit{cur}$;\\
{\tt 4}:\hspace*{0.4truecm} $\mathit{cur} := \mathit{cur}\rightarrow\next$;\\
\hspace*{0.34truecm} {\tt od}\\
{\tt 5}:\hspace*{0.13truecm} {\texttt new} $\mathit{tmp}$;\\
{\tt 6}:\hspace*{0.11truecm} $\mathit{tmp}\rightarrow\data$ := $\mathit{key}$;\\
{\tt 7}:\hspace*{0.11truecm} $\mathit{tmp}\rightarrow\next$ := $\mathit{cur}$;\\
{\tt 8}:\hspace*{0.13truecm} {\tt if} ($\mathit{prev}$ $!=$ $\nil$) {\tt then}\\
{\tt 9}:\hspace*{0.4truecm} $\mathit{prev}\rightarrow\next$ := $\mathit{tmp}$;\\
\hspace*{0.34truecm} {\tt else}\\
{\tt 10}: \hspace*{0.19truecm} $\mathit{head}$ := $\mathit{tmp}$;\\
\hspace*{0.34truecm} {\tt fi}\\

\columnbreak


\begin{picture}(177,49)(0,-49)
\setlength{\unitlength}{0.32mm}

\node[Nw=9.0,Nh=9.0,NLangle=130.0,NLdist=19.0,ilength=10.0,iangle=137.0,flength=8.0,Nmarks=if](n0)(16.0,-19.76){\scriptsize $\mathit{head}$}
\nodelabel[NLangle=0.0](n0){\scriptsize $2$}

\node[Nw=9.0,Nh=9.0,NLangle=130.0,NLdist=17.0,ilength=10.0,Nmarks=i](n1)(72.0,-19.76){\scriptsize $\mathit{prev}$}
\imark[ilength=10.0,iangle=137.0](n1)
\nodelabel[NLangle=0.0](n1){\scriptsize $6$}

\node[Nw=9.0,Nh=9.0,NLangle=130.0,NLdist=17.0,ilength=10.0,iangle=137.0,flength=8.0,Nmarks=if](n2)(96.0,-19.76){\scriptsize $\mathit{cur}$}
\nodelabel[NLangle=0.0](n2){\scriptsize $9$}

\node[Nw=9.0,Nh=9.0,NLangle=180.0,NLdist=22.0,Nmarks=i,ilength=8.0](n3)(84.0,-39.76){\scriptsize $\mathit{tmp}$}
\nodelabel[NLangle=0.0](n3){\scriptsize $8$}

\nodelabel[NLangle=211.0,NLdist=250.0](n3){\large (a)}
\nodelabel[NLangle=180.0,NLdist=130.0](n3){\large (b)}
\nodelabel[NLangle=224.0,NLdist=185.0](n3){\large (c)}
\nodelabel[NLangle=260.0,NLdist=130.0](n3){\large (d)}

\node[Nw=9.0,Nh=9.0,NLangle=130.0,NLdist=18.0,ilength=10.0,iangle=135.0,Nmarks=i](n4)(147.99,-20.0){\scriptsize $\nil$}
\imark[ilength=8.0](n4)
\nodelabel[NLangle=0.0](n4){\scriptsize $\$$}

\node[Nw=9.0,Nh=9.0,NLangle=0.0](n5)(167.99,-19.76){\scriptsize $\$$}
\nodelabel[NLdist=8.0](n5){}

\node[Nw=9.0,Nh=9.0,NLangle=180.0,NLdist=20.0,Nmarks=i,ilength=8.0](n6)(159.99,-39.76){\scriptsize $\mathit{key}$}
\nodelabel[NLangle=0.0](n6){\scriptsize $8$}

\drawedge[dash={2.0 2.0 2.0 3.0}{0.0},AHnb=0,sxo=15.0,exo=-15.0](n0,n1){ }

\drawedge(n1,n2){}

\drawedge[dash={2.0 2.0 2.0 3.0}{0.0},AHnb=0,sxo=15.0,exo=-15.0](n2,n4){ }

\drawedge(n4,n5){}

\drawedge(n6,n5){}

\drawedge(n3,n2){}

\drawline[AHnb=0, linegray=0.5, linewidth=2](-69,0)(-69,-187)

\drawline[AHnb=0, linegray=0.5, linewidth=2](42,-58)(42,-187)

\drawline[AHnb=0, linegray=0.5, linewidth=2](-70,-58)(180,-58)

\end{picture}


\begin{picture}(35,5)(15,-37)
\setlength{\unitlength}{0.4mm}

\node[Nfill=y,fillcolor=black,NLangle=0.0,NLdist=11.0,Nw=4.0,Nh=4.0,Nmr=2.0](n0)(27.76,-12.0){\scriptsize $(b,\$)$}

\node[Nfill=y,fillcolor=black,NLangle=180.0,NLdist=16.9,fangle=270.0,flength=6.0,Nmarks=f,Nw=4.0,Nh=4.0,Nmr=2.0](n1)(19.76,-28.0){\scriptsize $(\{\nil\},\$)$}

\drawedge(n0,n1){}

\node[Nfill=y,fillcolor=black,NLangle=0.0,NLdist=18.0,Nw=4.0,Nh=4.0,Nmr=2.0](n2)(35.76,-28.0){\scriptsize $(\{\mathit{key}\},8)$}

\drawedge(n0,n2){}

\node[Nfill=y,fillcolor=black,NLangle=180.0,NLdist=17.8,ilength=6.0,iangle=90.0,Nmarks=i,Nw=4.0,Nh=4.0,Nmr=2.0](n3)(19.76,-52.0){\scriptsize $(\{\mathit{cur}\},9)$}

\drawedge[dash={2.0 2.0 2.0 3.0}{0.0},AHnb=0,syo=-7.0,eyo=7.0](n1,n3){ }

\node[Nfill=y,fillcolor=black,NLangle=0.0,NLdist=18.5,Nw=4.0,Nh=4.0,Nmr=2.0](n5)(31.76,-64.0){\scriptsize $(\{\mathit{tmp}\},8)$}

\node[Nfill=y,fillcolor=black,NLangle=180.0,NLdist=19.0,flength=6.0,fangle=270.0,Nmarks=f,Nw=4.0,Nh=4.0,Nmr=2.0](n7)(19.76,-64.0){\scriptsize $(\{\mathit{prev}\},6)$}

\drawedge(n3,n5){}

\drawedge(n3,n7){}

\node[Nfill=y,fillcolor=black,NLangle=0.0,NLdist=19.3,ilength=6.0,iangle=90.07,Nmarks=i,Nw=4.0,Nh=4.0,Nmr=2.0](n9)(19.76,-88.0){\scriptsize $(\{\mathit{head}\},2)$}

\drawedge[dash={2.0 2.0 2.0 3.0}{0.0},AHnb=0,syo=-8.0,eyo=8.0](n7,n9){ }

\end{picture}


\begin{picture}(35,5)(-25,-42.5)
\setlength{\unitlength}{0.4mm}

\node[Nfill=y,fillcolor=black,NLangle=0.0,NLdist=16.0,Nw=4.0,Nh=4.0,Nmr=2.0](n0)(27.76,-12.0){\scriptsize $(b,-,\$)$}

\node[Nfill=y,fillcolor=black,NLangle=180.0,NLdist=23.0,fangle=270.0,flength=6.0,Nmarks=f,Nw=4.0,Nh=4.0,Nmr=2.0](n1)(19.76,-28.0){\scriptsize $(\{\nil\},-,\$)$}

\drawedge(n0,n1){}

\node[Nfill=y,fillcolor=black,NLangle=0.0,NLdist=24.0,Nw=4.0,Nh=4.0,Nmr=2.0](n2)(35.76,-28.0){\scriptsize $(\{\mathit{key}\},-,8)$}

\drawedge(n0,n2){}

\node[Nfill=y,fillcolor=black,NLangle=180.0,NLdist=24.0,ilength=6.0,iangle=90.0,Nmarks=i,Nw=4.0,Nh=4.0,Nmr=2.0](n3)(19.76,-52.0){\scriptsize $(\{\mathit{cur}\},{\bf y_2},9)$}

\drawedge[dash={2.0 2.0 2.0 3.0}{0.0},AHnb=0,syo=-7.0,eyo=7.0](n1,n3){ }

\node[Nfill=y,fillcolor=black,NLangle=0.0,NLdist=24.0,Nw=4.0,Nh=4.0,Nmr=2.0](n5)(31.76,-64.0){\scriptsize $(\{\mathit{tmp}\},-,8)$}

\node[Nfill=y,fillcolor=black,NLangle=180.0,NLdist=25.0,flength=6.0,fangle=270.0,Nmarks=f,Nw=4.0,Nh=4.0,Nmr=2.0](n7)(19.76,-64.0){\scriptsize $(\{\mathit{prev}\},{\bf y_1},6)$}

\drawedge(n3,n5){}

\drawedge(n3,n7){}

\node[Nfill=y,fillcolor=black,NLangle=0.0,NLdist=24.0,ilength=6.0,iangle=90.07,Nmarks=i,Nw=4.0,Nh=4.0,Nmr=2.0](n9)(19.76,-88.0){\scriptsize $(\{\mathit{head}\},-,2)$}

\drawedge[dash={2.0 2.0 2.0 3.0}{0.0},AHnb=0,syo=-8.0,eyo=8.0](n7,n9){ }

\end{picture}

\end{multicols}
\end{minipage}
}
\caption{ {\bf (a)} {\em sorted list-insert} program $P$; {\bf (b)} shows  a $P$ configuration at program counter $8$; {\bf (c)} is the heap skinny-tree associated to (b); {\bf (d)} is a valuation tree of (c).}\label{fig:prog}
\end{figure}

\vspace{-0.1cm}
\paragraph{\bf Heap skinny-trees.} Let $\PV$ be the set of pointer variables of a program $P$ and $\Sigma=2^\PV$ (let us denote the empty set with a blank symbol $b$).
 We associate with each $P$ configuration $C=\tuple{\pc, \heap, \pval, \dval}$ with $H=(\Loc, \next, \data)$, the $(\Sigma\times \dom)$-labelled graph $\HC=(T,\lambda)$ whose nodes are those of $\Loc$, and where $(u,v)$ is an edge of $T$ iff $\next(v)=u$ (essentially we reverse all $\next$ edges). From the definition of program configurations it is easy to see that $T$ is a $k$-skinny tree where $k=|\PV|$. The labelling function $\lambda:\Loc \mapsto(\Sigma\times \dom)$ is defined as follows: for every $u\in \Loc$, $\lambda(u)=(S,d)$ where $S$ is the set of all pointer variables $p$ such that $\pval(p)=u$, and $d=\data(u)$. We call $\HC$ the {\em heap skinny-tree} of $C$.

In general heap skinny-trees can be logically characterized as follows.

\begin{definition}[{\sc Heap Skinny-Trees}]
A {\em heap skinny-tree} over a set of pointer variables $\PV$ (with {\em \nil}$\in\PV $) and data domain $\dom$, is a $(\Sigma\times\dom)$-labelled $k$-skinny tree $(T,\lambda)$ with $\Sigma=2^\PV$ and $k=|\PV|$, such that:
\vspace{-0.1cm}
\begin{itemize}
\item for every leaf $v$ of $T$, $\lambda(v)=(S,d)$ where $S\not=\emptyset$;

\item for every pointer variable $p\in\PV$, there is a unique node $v$ of $T$ such that $\lambda(v)=(S,d)$ with $p\in S$;
\item the node $v$ of $T$ such that  $\lambda(v)=(S,d)$ and {\em \nil}$\in S$ is one of the childen of the root of $T$. \qed
\end{itemize}
\end{definition}

Figure~\ref{fig:prog}(c) shows the heap skinny-tree corresponding to the program configuration of Figure~\ref{fig:prog}(b). Note that though the program handles a singly linked list, in the intermediate operations we can get trees. However they are special trees with bounded branching.
%
This example illustrates that program configurations of list manipulating programs naturally correspond to heap skinny-trees. It also motivates why we need to extend automata over words introduced in~\cite{CAVQDA} to quantified data automata over skinny-trees.
We now define valuation trees.

\paragraph{\bf Valuation trees.} Let us fix a finite set of {\em universal} variables $Y$. A {\em valuation tree} over $Y$ of a heap skinny-tree $\HC$ is a $(\Sigma \times (Y \cup
\{\yblank\}) \times \dom)$-labelled tree  obtained from $\HC$ by adding an element from the set $Y \cup\{\yblank\}$ to the label, in which every element in $Y$ occurs exactly once in the tree. We use  the symbol `$\yblank$' at a node $v$ if there is no variable from $Y$ labelling  $v$.
A valuation tree corresponding to the heap skinny-tree of Figure~\ref{fig:prog}(c) is shown in Figure~\ref{fig:prog}(d).


\begin{definition}[{\sc Quantified Skinny-Tree Data Automata}]
A {\em quantified skinny-tree data automaton} (\QSDA) over a set of pointer variables $\PV$ (with $|\PV|=k$), a data domain $\dom$, a set of universal variables $Y$, and a set of data formulas $F$ over $\dom$, is a tuple $\A = (Q, \Pi, \Delta, \mathcal{T}, f)$ where:
\begin{itemize}
\item $Q$ is a finite set of states;
\item $\Pi$ = $\Sigma \times \widehat{Y}$ is the alphabet where $\Sigma = 2^{PV}$ and $\widehat{Y} = Y \cup \{-\}$;
\item $\Delta= (\Delta_0, \Delta_1, \ldots, \Delta_k)$ where, for every $i\in[1,k]$,  $\Delta_i : (Q^i \times \Pi) \mapsto Q$ defines a (deterministic) transition relation; 
\item $\mathcal{T}: Q \rightarrow 2^{PV \cup Y} $ is the type associated with every state $q \in Q$; 
\item $f: Q \mapsto F$ is a final-evaluation.\qed 
\end{itemize}
\end{definition}

A valuation tree $(T,\lambda)$ over $Y$ of a program $P$, where $N$ is the set of nodes of $T$, is {\em recognized}  by a \QSDA\ $\A$ if there exists a node-labelling map $\run:N\mapsto Q$ that associates each node of $T$ with a state in $Q$ such that for each node $t$ of $T$ with $\lambda(t)=(S,y,d)$ the following holds (here $\lambda'(t) = (S,y)$ is obtained by projecting out the data values from $\lambda(t)$):
\begin{itemize}
\item if $t$ is a leaf then $\Delta_0(\lambda'(t))=\run(t)$ and $(\mathcal{T}(\run(t))\cap \PV)\not=\emptyset$.

\item if $t$ is an internal node, with sequence of children $t_1,t_2, \ldots, t_i$ then
\begin{itemize}
\item $\Delta_i\left(\, (\run(t_1),\ldots, \run(t_i) ),\, \lambda'(t)\, \right)=\run(t)$;

 \item $S \cap \mathcal{T}(\run(t_j))=\emptyset$ and  $y\notin \mathcal{T}(\run(t_j))$, for every $j\in[1,i]$;
 \item $\mathcal{T}(\run(t))=S \cup \{y\} \cup \left(\bigcup_{j\in[1,i]} \mathcal{T}(\run(t_j))\right)$ if $y \in Y$. Otherwise if $y = -$ then $\mathcal{T}(\run(t))=S \cup \left(\bigcup_{j\in[1,i]} \mathcal{T}(\run(t_j))\right)$.

\end{itemize}
\item if $t$ is the root then $\mathcal{T}(\run(t))=(\PV\cup Y)$ and the formula $f(\run(t))$, obtained by replacing all occurrences of terms $y \rightarrow \data$ and $p \rightarrow \data$ with their corresponding data values in the valuation tree,
holds true.

\end{itemize}

A \QSDA\ can be thought as a \emph{register} automaton that reads a valuation tree in a bottom-up fashion and stores the data at the positions evaluated for $Y$ and locations pointed by elements in $\PV$, and checks whether the formula associated to the state at the root  holds true by instantiating the data values in the formula with those stored in the registers. Furthermore, the role of  map $\mathcal{T}$ is that of enforcing that each element in $\PV\cup Y$ occurs exactly once in the valuation tree.

A \QSDA\ $\A$ {\em accepts} a heap skinny-tree $\HC$ if $\A$ recognizes all valuation trees of $\HC$. The {\em language} accepted by $\A$, denoted $L(\A)$, is the set of all heap skinny-trees $\HC$ accepted by $\A$. A language $\cal L$ of heap skinny-trees is {\em regular} if there is a \QSDA\ $\A$ such that ${\cal L}=L(\A)$.
Similarly, a language $\cal L$ of valuation trees is {\em regular} if there is a \QSDA\ $\A$ such that ${\cal L}=\Lval(\A)$, where $\Lval(\A)$ is the set of all valuation trees recognized by $\A$.



QSDAs are a generalization of {\em quantified data automata} introduced in~\cite{CAVQDA} that handle only lists as opposed to QSDAs that handle skinny-trees.
We now introduce various characterizations of \QSDA s which are used later in the paper.

\paragraph{\bf Unique minimal \QSDA.} In~\cite{CAVQDA} the authors show that it is not possible to have a unique minimal quantified data automaton over words (with respect to the number of states) which accepts a given language over linear heap configurations. The proof gives a set of heap configurations over a linear heap-structure which is accepted by two different automata having the same number of states. Since QSDAs are a generalization of quantified data automata, the same counter-example language holds for \QSDA s.
However, under the assumption that all data formulas in $F$ are pairwise non-equivalent, there does exist a canonical automaton on the level of valuation trees. In \cite{CAVQDA}, the authors prove the canonicity of quantified data automata, and their result extends to \QSDA s in a straight forward manner.

\begin{theorem}
\label{thm-canonical}
For each \QSDA\ $\A$ there is a unique minimal \QSDA\ $\A'$ such that
$\Lval(\A)=\Lval(\A')$.
\end{theorem}

We give some intuition behind the proof of Theorem~\ref{thm-canonical}.
First, we introduce a central concept called \emph{symbolic trees}. A symbolic tree is a $(\Sigma \times (Y \cup \{-\}))$-labelled tree that records the positions of the universal variables and the pointer variables, but does not contain concrete data values (hence the word symbolic).
A valuation tree can be viewed as a symbolic tree augmented with data values at every node in the tree. There exists a unique tree automaton over the alphabet $\Pi$ that accepts a given regular language over symbolic trees. It can be shown that if the set of formulas in $F$ are pair-wise non-equivalent,
then each state $q$  in the tree automaton, at the root, can be decorated with a unique data formula $f(q)$ which extends the symbolic trees with data values such that the corresponding valuation trees are in the given language.

Hence, a language of valuation trees can be viewed as a function that
maps each symbolic tree to a uniquely determined formula, and a \QSDA\ can be viewed as a
Moore machine (an automaton with output function on states) that
computes this function.
This helps us separate the structure of valuation trees (the height of the trees, the cells the pointer variables point to) from the data contained in the nodes of the trees. We formalize this notion by introducing \emph{formula trees}.


\paragraph{\bf Formula trees.}
A \emph{formula tree} over pointer variables $PV$, universal variables $Y$ and a set of data formulas $F$ is a tuple of a $\Sigma \times (Y \cup \{\yblank\})$-labelled tree (or in other words a symbolic tree) and a data formula in $F$ such that if we extend the tree with data values which satisfy the formula, we get a valuation tree.
For a QSDA which captures a universally quantified property of the form $\bigwedge_i \forall y_1 \ldots y_\ell. Guard_i \Rightarrow Data_i$, the symbolic tree component of the formula tree corresponds to guard formulas like $Guard_i$ which express structural constraints on the pointers pointing into the valuation tree.
The data formula in the formula trees correspond to $Data_i$ which express the data values with which a symbolic tree (read $Guard_i$) can be extended so as to get a valuation tree accepted by the QSDA.
In our running example, consider a QSDA with a formula tree which has the same symbolic tree as the valuation tree in Figure~\ref{fig:prog}(d) (but without the data values in the nodes) and a data-formula $\varphi = y_1\rightarrow\data \leq y_2\rightarrow\data \wedge y_1\rightarrow\data < key \wedge y_2\rightarrow\data \geq key$.
This formula tree represents all valuation trees (including the one shown in Figure~\ref{fig:prog}(d)) which extend the symbolic tree with data values which satisfy $\varphi$.

%


By introducing formula trees we explicitly take the view of a \QSDA\ as an automaton that reads symbolic trees and outputs data formulas.
We say a formula tree $(t, \varphi)$ is accepted by a \QSDA\ $\mathcal A$ if $\mathcal A$ reaches the state $q$ after reading $t$ and $f(q) = \varphi$.
Given a \QSDA\ $\A$, the language of valuation trees accepted by $\A$ gives an equivalent language of formula trees accepted by $\A$ and vice-versa. We denote the set of formula trees accepted by $\A$ as $L_f(\A)$. A language over formula trees is called regular if there exists a QSDA accepting the same language.
\begin{theorem}\label{thm-formula-trees}
For each \QSDA\ $\A$ there is a unique minimal \QSDA\ $\A'$ that accepts the
same set of formula trees.
\end{theorem}

\section{\QSDA s as an Abstract Domain}

In the previous section we introduced quantified skinny-tree data automata as an automaton model for expressing universally quantified properties over heap skinny-trees. In this section, we first show that QSDAs form a lattice and then formalize the correspondence, by establishing an abstraction function and a concretization function, between a set of heap skinny-trees and \QSDA s.

Given a set of pointer variables $\PV$ and universal variables $Y$, let $\mathcal{Q}_F$ be the set of all \QSDA\ over a set of data formulas $F$. Clearly $\mathcal{Q}_F$ is a partially-ordered set where the most natural partial order is the set-inclusion over the language of \QSDA s.
However checking whether a pair of \QSDA s are ordered with respect to this partial order is undecidable. Since \QSDA s generalize the quantified data automata over words~\cite{CAVQDA}, the undecidability follows from the fact that quantified data automata (as well as \QSDA s) can express quantified invariants such that checking the validity of such invariants is undecidable.

So, we consider a new partial-order on \QSDA s which is decidable, allows us to define a unique \emph{least upper bound} for every pair of \QSDA s and finally show that QSDAs form a lattice. To accomplish this, let us first assume that the set of formulas $F$ parameterizing \QSDA s  form a lattice $\F = (F, \sqsubseteq_\F, \lub_\F, \glb_\F, \emph{false}, \emph{true})$ where $\sqsubseteq_\F$ is the partial-order on the data-formulas, $\lub_\F$ and $\glb_\F$ are the least upper bound and the greatest lower bound and \emph{false} and \emph{true} are formulas required to be in $F$ and correspond to the bottom and the top elements of the lattice, respectively. Also, we assume  that whenever $\alpha \sqsubseteq_\F \beta$ then $\alpha \Rightarrow \beta$. Furthermore, we assume that any pair of formulas in $F$ are non-equivalent.
For a  logical domain as ours, this can be achieved by having a canonical representative for every set of equivalent formulas.

Now if we view a QSDA as a mapping from symbolic trees to formulas in $\F$, we can define a new partial-order relation on \QSDA s as follows. We say $\A_1 \sqsubseteq \A_2$ if $L_f(\A_1) \subseteq L_f(\A_2)$, which means that for every symbolic tree $t$ if $(t, \varphi_1) \in L_f(\A_1)$ and $(t, \varphi_2) \in L_f(\A_2)$ then $\varphi_1 \sqsubseteq_\F \varphi_2$.
Note that, whenever $\A_1 \sqsubseteq \A_2$ implies that $L(\A_1) \subseteq L(\A_2)$.
Also, with respect to this new partial order, we can show that \QSDA s form a complete lattice $(\mathcal{Q_F}, \sqsubseteq, \lub, \glb, \bot, \top)$ where the join of the two automata $\A_1$ and $\A_2$ maps the symbolic tree $t$ to the unique formula $\varphi_1 \lub_\F \varphi_2$. Similarly, the meet of the automata $\A_1$ and $\A_2$ maps the tree $t$ to $\varphi_1 \glb_\F \varphi_2$.
The bottom element in the lattice $\mathcal{Q_F}$ is the \QSDA\ which maps every symbolic tree to $\false$ and the top element is the \QSDA\ which maps every symbolic tree to the formula $\true$.

We now define an abstraction function $\alpha: \HC \rightarrow \mathcal{Q_F}$ and a concretization function $\gamma: \mathcal{Q_F} \rightarrow \HC$ such that $(\HC, \alpha, \gamma, \mathcal{Q_F})$ form a Galois-connection.
Note that, abstract interpretation~\cite{cc77} requires that the abstraction function $\alpha$ maps a concrete element (a language of heap skinny-trees) to a unique element in the abstract domain and that $\alpha$ be surjective; similarly $\gamma$ should be an injective function.
Also note that given a regular language of heap skinny-trees there might be several \QSDA s accepting that language. In such a case defining a surjective function $\alpha$ is not possible.
Therefore,
%
we first restrict ourselves to a set of \QSDA s in $\mathcal{Q_F}$ where each QSDA accepts a different language.
%
%
Under this assumption, we define an $\alpha$ and a $\gamma$ as follows: for a set of heap configurations $\HC$, $\alpha(\HC) = \bigsqcap \{\A ~ | ~ \HC \subseteq L(\A)\} $
and $\gamma(\A) = \{\heap ~ | ~ \heap \in L(\A) \}$.
%
%
Note that both $\alpha$ and $\gamma$ are order-preserving; $\alpha$ is surjective and $\gamma$ is an injective function. Also for a set of heap configurations $\mathcal{H}$, $\mathcal{H} \subseteq \gamma(\alpha(\mathcal{H}))$ and for a QSDA $\A$, $\A = \alpha(\gamma(\A))$.
Hence $(\HC, \alpha, \gamma, \mathcal{Q_F})$ form a Galois-connection.
\begin{theorem}
Let $(\HC, \subseteq)$ be the set of all heap skinny-trees and $(\mathcal{Q_F}, \sqsubseteq)$ be the set of \QSDA s (accepting pairwise inequivalent languages) over data formulas $\F$, then $(\HC, \alpha, \gamma, \mathcal{Q_F})$ form a Galois-connection.
\end{theorem}


\section{Abstract Semantics over \QSDA s}
\label{qdas}


In the previous section we established a Galois-connection between a set of heap skinny-trees and \QSDA s. Here, we describe an abstract transformer over \QSDA s which soundly over-approximates the concrete semantics of the programming language. This provides a way to compute the semantics of a program over an abstract domain consisting of \QSDA s.


We first show that it is not possible to capture the most-precise concrete transformer on \QSDA s.
A QSDA expresses universally quantified properties over heap trees, of the form $\forall y_1 \ldots y_\ell. \psi$ where $\psi$ is a quantifier-free formula over the pointer variables $\PV$, the universal variables $Y$ and the data value at the locations pointed to by these variables.
Given a \QSDA\ $\A$, the concrete transformer $F^\natural$ guesses a pre-state accepted by $\A$ (which involves existential quantification), and then constrains the post-state with respect to this guessed pre-state according to the semantics of the statement.
%
For instance, consider the statement $p_i := p_j$. Given a \QSDA\ accepting a universally quantified property $\forall y_1 \ldots y_\ell. \psi$, its strongest post-condition with respect to this statement is the formula: $\exists p_i'. \forall y_1 \ldots y_\ell. \psi[p_i/p_i'] ~\wedge~ p_i = p_j$.
%
%
\noindent Note that, an interpretation of the existentially quantified variable $p_i'$ in a model of this formula gives the location node pointed to by variable $p_i$ in the pre-state, 
such that the formula $\forall y_1 \ldots y_\ell. \psi$ was satisfied by the pre-state.
However it is not possible to express these precise post-conditions, which are usually of the form $\exists^* \forall^* \psi$, in our automaton model.
So we abstract these precise post-conditions by a \QSDA\ which semantically moves the existential quantifiers inside the universally quantified prefix, where they can be eliminated.
In the above example, the abstract post-condition \QSDA\ guesses a  position for the pointer variable $p_i$ 
for every valuation of the universal variables,  such that the valuation tree augmented with this guessed position 
is accepted by the precondition \QSDA. More generally, the abstract transformer computes the most precise post-condition over the language of valuation trees 
accepted by a \QSDA, instead of computing the precise post-condition over the language of heap skinny-trees.
In fact, we go beyond valuation trees to formula trees; the abstract transformer evolves the language of formula trees accepted by a \QSDA\ by tracking the precise set of symbolic trees to be accepted in the post-\QSDA\ and their corresponding data formulas.

%

\begin{table*}[th]
\centering
\captionsetup{font=small}
\scalebox{0.85}{
\begin{tabular}{| l | l |}

\hline
{\bf Statements} & ~~~~{\bf Abstract Transformer $F_f^\sharp$ on a regular language over formula trees}\\\hline\hline
$p_i := nil$ & $\lambda L_f.~ \big\{(t', \varphi') ~|~\varphi': \bigsqcup \{ \exists d. \varphi[p_i\rightarrow\data/d] ~ | ~ (t, \varphi) \in L_f, update(t, p_i := nil, t') \}  \big\}$ \\\hline

$p_i := p_j$ & $\lambda L_f.~ \big\{(t', \varphi') ~|~\varphi': (p_i\rightarrow\data = p_j\rightarrow\data) \sqcap ~ $\\
& $~~~~~~~~~~~~~~~~~~~~~~~~~\bigsqcup \{ \exists d. \varphi[p_i\rightarrow\data/d] ~ | ~ (t, \varphi) \in L_f, update(t, p_i := p_j, t') \}  \big\}$ \\\hline

$p_i := p_j\rightarrow\next$ & $\lambda L_f.~ \big\{(t', \varphi') ~|~\varphi': \bigsqcup \{ \exists d. \varphi[p_i\rightarrow\data/d] ~ | ~ (t, \varphi) \in L_f, update(t, p_i := p_j\rightarrow\next, t') \} $ \\
& $ ~~~~~~~~~~~~~~~~~~~~~~~~~\sqcap  (p_i\rightarrow\data = v\rightarrow\data), v \in label(t', p_i)  \big\}$ \\\hline

$p_i\rightarrow\next := nil$ & $\lambda L_f.~ \big\{(t', \varphi') ~|~\varphi': \bigsqcup \{ \varphi ~ | ~ (t, \varphi) \in L_f, update(t, p_i\rightarrow\next := nil, t') \}  \big\}$ \\\hline

$p_i\rightarrow\next := p_j$ & $\lambda L_f.~ \big\{(t', \varphi') ~|~\varphi': \bigsqcup \{ \varphi ~ | ~ (t, \varphi) \in L_f, update(t, p_i\rightarrow\next := p_j, t') \}  \big\}$ \\\hline

$p_i\rightarrow\data := $ & $\lambda L_f.~ \big\{(t', \varphi') ~|~ \varphi': \exists d. \varphi[v_1\rightarrow\data/d,\ldots,v_\ell\rightarrow\data/d] $ \\
$~~~~~~~~data\_expr$ & $~~~~~~~~~~~~~~~~~~~~~~~~~\sqcap \bigsqcap \{ v\rightarrow\data = data\_expr ~|~ v \in V\}$, \\
&~~~~~~~~~~~~~~~~~~~~~~~~~~~~~~~~~~~~~~~~~~~~~~~~~~~~$V =\{v_1,\ldots,v_\ell\}= label(t', p_i), (t', \varphi) \in L_f \big\}$\\\hline

$\textit{assume } \psi_{\mathit{struct}}$ & $\lambda L_f.~ \big\{(t', \varphi') ~|~ (t', \varphi') \in L_f, ~
t' \models \psi_{\mathit{struct}} \big\}$ \\\hline

$\textit{assume } \psi_{\mathit{data}} $ & $\lambda L_f.~ \big\{(t', \varphi') ~|~ \varphi': \varphi \sqcap \psi_{\mathit{data}}, (t', \varphi) \in L_f \big\}$ \\\hline

$\textit{new } p_i$ & $
\lambda L_f. \big\{ (t', \varphi') ~ | ~ \varphi':  (y\rightarrow\data = p_i\rightarrow\data) \sqcap$ \\
& $~~~~~~~~~ \bigsqcup \{ \exists d_1 d_2. \varphi[p_i\rightarrow\data/d_1, y\rightarrow\data/d_2] ~|~ (t, \varphi) \in L_f, update(t, new^y~p_i, t') \},$\\
& $~~~~~~~~~~~~~~~~~~y \in Y \cup \{-\}  \big\}$ \\\hline

\end{tabular}}
\caption{Abstract Transformer $F_f^\sharp$. The abstract transformer $F^\sharp = \lambda \mathcal{A}. \mathcal{A'}$ where $\mathcal{A'}$ is the unique minimal \QSDA\ such that $L_f(\mathcal{A'}) = (F_f^\sharp) ~L_f(\mathcal{A})$. The predicate $update$ and the set $label$ are defined below.}
\label{tbl-abstract-transformer}
\end{table*}

Table~\ref{tbl-abstract-transformer}\footnote{The abstract transformer defined in Table~\ref{tbl-abstract-transformer} assumes that there are no memory errors in the program. It can be extended to handle memory errors.} gives the abstract transformer $F_f^\sharp$ which takes a regular language over formula trees $L_f$ and gives, as output, a set of formula trees. 
We know from Theorem~\ref{thm-formula-trees} that for any regular set of formula trees there exists a unique minimal \QSDA\ that accepts it. We show below (see Lemma~\ref{lemma-regular}) that for a \QSDA\ $\mathcal{A}$, the language over formula trees given by $(F_f^\sharp) ~L_f(\mathcal{A})$ is regular. Hence, we can define the abstract transformer $F^\sharp$ as $F^\sharp = \lambda \mathcal{A}. \mathcal{A'}$ where $\mathcal{A'}$ is the unique minimal \QSDA\ such that $L_f(\mathcal{A'}) = (F_f^\sharp) ~L_f(\mathcal{A})$.


In Table~\ref{tbl-abstract-transformer}, $label(t, p_i)$ is the set of pointer and universal variables 
which label the same node in $t$ as variable $p_i$. 
The predicate $update(t, \mathit{stmt}, t')$ is true if symbolic trees $t$ and $t'$ are related such that the execution of statement $\mathit{stmt}$ updates precisely the symbolic tree $t$ to $t'$.
As an example, the abstract transformer for the statement $p_i := nil$ in the first row of Table~\ref{tbl-abstract-transformer} states that the post-\QSDA\ maps the symbolic tree $t'$ to the data-formula $\varphi'$ where $\varphi'$ is the join of all formulas of the form $\exists d. \varphi[p_i\rightarrow\data/d]$ where $\varphi$ is the data-formula associated with symbolic tree $t$ in the pre-\QSDA\ such that $update(t, p_i := nil, t')$ is true.

We now briefly describe the predicate $update(t, new^y~p_i, t')$ which is used in the definition of the transformer for the \textit{new} statement and is slightly more involved.
The statement $\textit{new } p_i$ allocates a new memory location. After the execution of this statement, pointer $p_i$ points to this allocated node. Besides, the universal variables also need to valuate over this new node apart from the valuations over the previously exisiting locations in the heap.
The superscript $y$ in the predicate $update(t, new^y~p_i, t')$ tracks the case when variable $y \in Y \cup \{-\}$ valuates over the newly allocated node.
Hence, if $update(t, new^y~p_i, t')$ holds true then the symbolic trees $t$ and $t'$ agree on the locations pointed to by all variables except $p_i$ and the universal variable $y$; both these variables point, in $t'$, to a new location $v$ which is not in $t$ and a new edge exists in $t'$ from the root to $v$.


From the construction in Table~\ref{tbl-abstract-transformer} it can be observed that given a language of valuation trees obtained uniquely from a language of formula trees, $F_f^\sharp$ applies the most-precise concrete transformer on each valuation tree in the language, and then constructs the smallest regular language of valuation trees (or equivalently formula trees) which approximates this set.
More precisely, for all formula trees $(t, \varphi) \in L_f(\A)$, the abstract transformer $F_f^\sharp$ applies the precise concrete transformer on the symbolic tree $t$ (only the structure with the valuations for universal variables) to obtain $t'$. And separately, it applies the precise concrete transformer on the data extensions of $t$, which is given by $\varphi$, to obtain the data formula  $\varphi'$ such that $(t', \varphi') \in (F_f^\sharp) L_f(\A)$.
However, note that reasoning over valuation/formula trees (and not heap skinny-trees) comes with a loss in precision.
To regain some of this lost precision, we define a function \emph{Strengthen} which takes a formula language $L_f$ and finds a smaller language over formula trees, which accepts the same set of heap trees. Here $t \downharpoonright_y$ stands for a $\Pi \backslash \{y\}$ -labelled tree which agrees with $t$ on the locations pointed to by all variables except $y$.
\begin{align}
\textit{Strengthen} =   \lambda y. \lambda L_f. \big\{ (t', &\varphi') ~ | ~ \varphi': \varphi'' \sqcap \phi, ~(t', \varphi'') \in L_f,  \nonumber \\
	  		& \phi: \bigsqcap \{ \exists d. \varphi[y\rightarrow\data/d] ~|~ (t, \varphi) \in L_f, t \downharpoonright_y = t' \downharpoonright_y \}   \big\} \nonumber
\end{align}
\noindent We now reason about the soundness of the operator \emph{Strengthen}. Fix a $y \in Y$. Consider a \QSDA\ $\A$ with a language over formula trees $L_f$ and consider all symbolic trees $t$ such that $t \downharpoonright_y = t' \downharpoonright_y$. This implies that
the trees $t$ have the pointer variables pointing to the same positions as $t'$
and have the same valuations for variables in $Y \backslash \{y\}$.
Since our automaton model has a universal semantics, any heap tree accepted by $\A$ should satisfy the data formulas annotated at the final states reached for every valuation of the universal variables. If we look at a fixed valuation for variables in $Y \backslash \{y\}$ (which is same as that in $t'$) and different valuations for $y$, any heap tree accepted should satisfy the formula $\exists d. \varphi[y\rightarrow\data/d]$ for all such $(t, \varphi) \in L_f$. Hence the \emph{Strengthen} operator can safely strengthen the formula $\varphi''$ associated with the symbolic tree $t'$ to $\varphi'' \sqcap \phi$.
Appendix~\ref{app-strengthen} shows that for a given universal variable $y$ and a regular language $L_f$, the language over formula trees (\textit{Strengthen}) $y$ $L_f$ is regular. The proof in fact constructs the \QSDA\ accepting the language (\textit{Strengthen}) $y$ $L_f(\A)$ for a \QSDA\ $\A$.
The abstract transformer $F_f^\sharp$ can be thus soundly strengthened by an application of \textit{Strengthen} at each step, for each variable $y \in Y$.

We now prove that the language over formula trees given by $(F_f^\sharp) L_f(\A)$ is a regular language for any \QSDA\ $\A$. This helps us to construct the abstract transformer $F^\sharp: \mathcal{Q_F} \rightarrow \mathcal{Q_F}$. And finally, we show that this abstract transformer is a sound approximation of the concrete transformer $F^\natural$.

\begin{lemma} \label{lemma-regular}
For a \QSDA\ $\mathcal{A}$, the language $(F_f^\sharp) ~L_f(\mathcal{A})$ over formula trees is regular.
\end{lemma}
\begin{proof}
We prove via construction. Given a \QSDA\ $\mathcal{A}$, we construct a \QSDA\ $\mathcal{A'}$ such that $L_f(\mathcal{A'}) = (F_f^\sharp) ~L_f(\mathcal{A})$.

Here we only give the construction of $\mathcal{A'}$ for the statement $p_i := nil$ (for other statements, see Appendix~\ref{app-construction}). The \QSDA\ $\mathcal{A'}$ simulates $\mathcal{A}$ on all the nodes, except a node $v$ labeled with pointer $p_i$ and the node labeled $nil$.
A tree accepted by $\mathcal{A'}$ does not read $p_i$ at node $v$; on the other hand $p_i$ is read at the $nil$ node.
Let $\mathcal{A} = (Q, \Pi, \Delta, \mathcal{T}, f)$, then $\mathcal{A'}$ is of the form $(Q, \Pi, \Delta', \mathcal{T'}, f')$. For every transition $\Delta(q_1, \ldots, q_j, \pi) =  q$ such that
$p_i$ and $nil$ are not present in the label $\pi$,  $\Delta'(q_1, \ldots, q_j, \pi) = q$.
However if $p_i$ is present in $\pi$ and $nil$ is not, then the corresponding transition in $\mathcal{A'}$ is $\Delta'(q_1, \ldots, q_j, \pi') = q$ where $\pi'$ is same as $\pi$ except it does not have $p_i$.
On the other hand, if $nil$ is present in $\pi$ and $p_i$ is not, then the transition $\Delta'(q_1, \ldots, q_j, \pi') = q$ where $\pi'$ is same as $\pi$ except for the presence of $p_i$.
The new type $\mathcal{T'}$ can be easily computed for every state in the automaton. The evaluation function $f'$ existentially quantifies out the data value of pointer $p_i$ i.e. for all states $q \in Q$, $f'(q) = \exists d. f(q)[p_i\rightarrow\data/d]$.
The transition relation $\Delta'$ thus constructed might need to be determinized to obtain $\mathcal{A'}$.
For a symbolic tree $t$ and formulas $\varphi_1, \ldots, \varphi_j$ such that $(t, \varphi_1), \ldots, (t, \varphi_j)$ belong to the language, the determinization procedure maps $t$ to the formula $\varphi_1 \sqcup \ldots \sqcup \varphi_j$.
It can be easily shown that the language of $\A'$ is $(F^\sharp)~L_f(\mathcal{A})$.

Appendix~\ref{app-construction} shows the construction of the automaton $\mathcal{A'}$ for other statements in our language. In this way, via construction, we prove that  the language $(F^\sharp)~L_f(\mathcal{A})$ over formula trees is regular.
\qed
\end{proof}

From Lemma~\ref{lemma-regular} and Theorem~\ref{thm-formula-trees} it follows that there exists a \QSDA\ $\mathcal{A'}$ such that $\mathcal{A'} = (F^\sharp) \mathcal{A}$. In fact the proof of Lemma~\ref{lemma-regular} constructs such an automaton $\mathcal{A'}$.
The monotonicity of $F^\sharp$ with respect to $F^\natural$ follows from the monotonicity of $F_f^\sharp$.
The soundness of $F^\sharp$ can be stated as the following theorem.

\begin{theorem}
The abstract transformer $F^\sharp$ is sound with respect to the concrete transformer $F^\natural$.
\end{theorem}
\begin{proof}
We prove the soundness of $F^\sharp$ by showing that  $F^\natural \circ \gamma \sqsubseteq \gamma \circ F^\sharp$.
Let us consider a QSDA $\A$ and a heap skinny-tree $\HC$ such that $\HC \in L(\A)$. Consider a statement $stmt$ such that $\HC$ gets transformed to $\HC'$ on the execution of $stmt$ i.e. $F^\natural_{stmt}(\HC) = \HC'$. We would like to prove that $\HC' \in L(F^\sharp(\A))$.

To prove this, consider a valuation of universal variables $Y$ over the nodes in $\HC'$. Let the corresponding symbolic tree be $t'$ and let the data values in $\HC'$ at the positions pointed to by $Y$ be given by $r': Y \rightarrow \mathbb{D}$. Let us assume that the QSDA $F^\sharp(\A)$ maps the symbolic tree $t'$ to the formula $\varphi'$. Then, we would like to prove that $r' \models \varphi'$. By arguing over all valuations over $\HC'$, this would prove that $\HC' \in L(F^\sharp(\A))$.

To prove that  $r' \models \varphi'$, fix a valuation of the universal variables $Y$ over the nodes in $\HC$ such that the corresponding symbolic tree $t$ satisifes $update(t, stmt, t')$ (for statements which do not modify the structure of the heap, $t = t'$). Let the data values at the positions pointed to by universal variables $Y$ in $\HC$ be given by $r$. Since $\HC \in L(\A)$, if $\A$ maps the symbolic tree $t$ to the formula $\varphi$ then $r \models \varphi$.
The abstract transformer $F_f^\sharp$ applies the precise concrete transformer, with respect to only the data values of the heap, to the formula $\varphi$ and over-approximates it to obtain $\varphi'$. From the monotonicity of the concrete transformer, this implies that $r' \models \varphi'$.
\qed
\end{proof}

Since $F^\sharp$ is both monotonic and sound, from the Knaster-Tarski theorem, the set of equations of the form $\psi = F^\sharp(\psi)$ expressing the abstract semantics of a program admit a least fix point solution, and this solution is a sound approximation of the concrete semantics of the program.

Note that the abstract transformer, in general, might require a powerset construction over the input QSDA, very similar to the procedure for determinizing a tree automaton. Hence the worst-case complexity of the abstract transformer is exponential in the size of the QSDA. However our experiments show that this worst-case is not achieved for most programs in practice.
\begin{theorem}
The abstract semantics of a program, computed with respect to the abstract transformer $F^\sharp$, is correct.
\end{theorem}

\section{Elastic Quantified Skinny-Tree Data Automata}
For a given set of pointer variables $\PV$ and universal variables $Y$, the QSDAs can be of arbitrarily large size. The number of QSDAs is not bounded and the computation of the abstract semantics of a program over QSDAs might not converge.
To remedy this problem, we identify a sub-class of QSDAs called elastic quantified  skinny-tree data automata (\EQSDA s).
Elastic QSDAs provide a mechanism to accelerate the fix-point computation over QSDAs. However, instead of choosing any acceleration mechanism, the elastic QSDAs were chosen keeping decidability of the invariants they express in mind.  A key property in the decidable fragment of $\Strand$ is that it cannot test whether two universally quantified variables are a bounded distance away.
We show in Section~\ref{sec-translation} that the invariants expressed by \EQSDA s fall in the decidable fragment of $\Strand$.
So \EQSDA s not only help in guaranteeing the convergence of the abstract semantics of a program, but also ensure that a program, if annotated with a set of assertions over logical formulas in $\Strand$, can be proved correct by validating those assertions in a decidable manner.

Let us denote 
the symbol $(b, \yblank) \in \Pi$ indicating that a position does not contain any variable by $\blank$.
A QSDA $A = (Q, \Pi, \Delta, \mathcal{T}, f)$ where $\Delta = (\Delta_0, \Delta_1,  \ldots, \Delta_k)$ is called elastic if each transition on $\blank$ in $\Delta_1$ is a self loop i.e. $\Delta_1(q_1, \blank) = q_2$ implies $q_1= q_2$.

We first show that the number of \EQSDA s is bounded for a fixed set $\PV$ and $Y$.
Recall that heap skinny-trees accepted by QSDAs require that the number of branching nodes in the skinny trees are bounded. So, the only infinity in the size of a skinny-tree is due to an unbounded number of $\blank$-labelled nodes which might occur along linear segments of the tree.
If we simulate an elastic QSDA on a skinny-tree accepted by it, all consecutively occurring $\blank$-labelled nodes along linear segments of the tree are labelled with the same state (due to the elasticity property).
To count the maximum number of states, in an \EQSDA, required to accept a language over heap trees, we might as well consider only those trees in the language which have no $\blank$-labelled nodes occurring along linear segments in the tree. For a given set $\PV$ and $Y$, the number of such skinny-trees and their sizes are bounded. This bounds the number of states in an \EQSDA\ which accepts any language over heap skinny-trees. This also proves that for a given $\PV$ and $Y$, the number of \EQSDA s are bounded.


We next show the  following
result that every QSDA $\A$ can be \emph{uniquely over-approximated}
by a language of valuation trees (or equivalently formula trees) that can be
accepted by an \EQSDA\ $\AEL$.
We will refer to this construction, which we outline below, as \emph{elastification}.
This result is an extension of the unique over-approximation result for quantified data automata over words~\cite{CAVQDA}.
Using this result, we can show that elastic QSDAs form a finite join semi-lattice and there exists a Galois-connection  $\tuple{\alpha^{el}, \gamma^{el}}$
between QSDAs and the set of \EQSDA s.
This lets us define an abstract transfomer over the abstract domain \EQSDA s such that the semantics of a program can be computed over \EQSDA s (it terminates)
in a sound  manner.

Let $\A = (Q, \Pi, \Delta, \mathcal{T}, f)$ be a QSDA  such that $\Delta = (\Delta_0, \Delta_1, \ldots, \Delta_k)$ and for a state $q$ let
$R_{\blank}(q):= \{q' \mid q' = q \text{ or } \exists q''. q'' \in R_{\blank}(q) \text{ and } \Delta_1(q'', \blank) = q' \}$
be the set of states reachable from $q$ by a (possibly empty) sequence
of $\blank$-unary-transitions.  For a set $S \subseteq Q$ we let
$R_{\blank}(S) := \bigcup_{q \in S} R_{\blank}(q)$.

The set of states of $\AEL$ consists of sets of states of $\A$
that
are reachable by the following transition function $\Delta^{el}$ (where $\Delta_i(S_1, \ldots, S_i,a)$ denotes the
standard extension of the transition function of $\A$ to sets of
states):
{\footnotesize
\begin{align}
\Delta_0^{\text{el}}(a) =& ~R_{\blank}(\Delta_0(a)) \nonumber \\
\Delta_1^{\text{el}}(S, a) =&
\begin{cases}
R_{\blank}(\Delta_1(S,a)) & \mbox{if } a \not= \blank \\
S  & \mbox{if } a = \blank \mbox{ and $\Delta_1(q,\blank)$ is defined  for some $q \in S$} \\
\mbox{undefined}  & \mbox{otherwise.} \\
\end{cases} \nonumber \\
\Delta_i^{\text{el}}(S_1, \ldots, S_i, a) =& ~R_{\blank}(\Delta_i(S_1, \ldots, S_i, a )) \mbox{ for } i \in [2, k] \nonumber
\end{align}
}
\noindent Note that this construction is similar to the usual powerset construction
except that in each step we apply the transition function of $\A$ to the
current set of states and take the $\blank$-closure. However, if the
input letter is $\blank$ on a unary transition, $\AEL$ loops on the current
set if a $\blank$-transition is defined for some state in the set.

It can be argued inductively, starting from the leaf states, that the type for all states in a set is the same. Hence we define the type of a set $S$ as the type of any state in $S$.
The final evaluation formula for a set is the least upper bound of the
formulas for the states in the set:
$f^{\text{el}}(S) = \bigsqcup_{q \in S}f(q)$.
We can now show that $\AEL$  is the
\emph{most precise over-approximation} of the language of valuation trees accepted by QSDA $\A$.

\begin{theorem}\label{thmelastification}
For every QSDA $\A$, the \EQSDA\ $\AEL$ satisfies
$\Lval(\A) \subseteq \Lval(\AEL)$, and
for every \EQSDA\ $\B$ such that $\Lval(\A) \subseteq \Lval(\B)$, 
 $\Lval(\AEL) \subseteq \Lval(\B)$ holds.
\end{theorem}

A proof of Theorem~\ref{thmelastification} is presented in Appendix~\ref{app-elastification} and is similar to the proof of Theorem~3 in~\cite{CAVQDA} for the case of words.
The above theorem can also be stated over a language of formula trees in the same way, that $\AEL$ is the most precise
over-approximation of the language of formula trees accepted by QSDA $\A$.

Using this result, we next show that \EQSDA s form a finite join semi-lattice $(\mathcal{Q_F}^{el}, \sqsubseteq, \lub, \bot, \top)$. The partial order on \EQSDA s is the same as the partial order on QSDAs but now restricted to elastic QSDAs.
For two \EQSDA s $\A_1$ and $\A_2$, the join $\A_1 \lub \A_2$ is the unique \EQSDA\ over-approximating $\A_1 \lub_{\mathcal{Q_F}} \A_2$. The bottom and the top elements are the \EQSDA s taking every symbolic tree to the formulas \emph{false} and \emph{true} respectively.
We can now view the space of \EQSDA s as an abstraction over  QSDAs.
The abstraction function $\alpha^{el}: \mathcal{Q_F} \rightarrow \mathcal{Q_F}^{el}$  takes a QSDA $\A$ to its unique over-approximating \EQSDA\ $\AEL$. The concretization function $\gamma^{el}: \mathcal{Q_F}^{el} \rightarrow \mathcal{Q_F}$ is the identity function which maps an \EQSDA\ to itself. Recall that we had already restricted $\mathcal{Q_F}$ to contain only those QSDAs which accepted different languages (over heap skinny-trees). Since $\mathcal{Q_F}^{el}$ is a sub-space of $\mathcal{Q_F}$, this restriction extends to it in a natural way. With this assumption it is easy to see that $\tuple{\alpha^{el}, \gamma^{el}}$ forms a Galois-connection.

Let us define the abstract transformer $F_{el}^\sharp: \mathcal{Q_F}^{el} \rightarrow \mathcal{Q_F}^{el} = \alpha_{el} \circ F^\sharp \circ \gamma_{el}$. The soundness of $F_{el}^\sharp$ follows from the soundness of $F^\sharp$ (and the fact that $\tuple{\alpha^{el}, \gamma^{el}}$ form a Galois-connection). Similarly its monotonicity follows from the monotonicity of $F^\sharp$ and the monotonicity of $\alpha_{el}$ and $\gamma_{el}$. The semantics of a program can be thus computed over the abstract domain $\mathcal{Q_F}^{el}$ as the least fix-point of a set of equations of the form $\psi = F_{el}^\sharp (\psi)$.
Since there are a bounded number of \EQSDA s for a given set of program variables $\PV$ and universal variables $Y$, this least fix-point computation terminates (modulo the convergence of the data formulas in the formula lattice $\F$ in which case termination can be achieved by defining a suitable widening operator on the data formula lattice).

\begin{theorem}
The abstract semantics of a program, computed with respect to the abstract transformer $F_{el}^\sharp$, is computable and is correct.
\end{theorem}

\subsection{From \EQSDA s to a Decidable Fragment of \Strand}
\label{sec-translation}

\EQSDA s introduced in  the previous section can express quantified data invariants over acyclic singly-linked data structures. In this section we show that \EQSDA s have a nice property that the quantified invariants expressed by them fall in a decidable fragment of first order logic, in particular the decidable fragment of \Strand. Hence, once the fix-point computation has converged, the invariants expressed by the \EQSDA s can be used to validate assertions in the program using decision procedures. In fact, the automaton model for \EQSDA s was chosen keeping in mind the decidability of the invariants expressed by them.

Given an \EQSDA\ $\A$ we would like to translate it to an equivalent formula $I$ such that the set of heap skinny-trees accepted by $\A$ corresponds to the program configurations which model $I$.
Recall that for any $k-$skinny-tree $\HC$ accepted by an \EQSDA, the number of branching nodes in $\HC$ is bounded by $k$. The invariants $I$ expressed by an \EQSDA\ are quantified formulas of the form $\exists b_1 \ldots b_k. \forall y_1 \ldots y_{\ell}. \varphi$ such that, in any model satisfying $I$, the existential variables $B = \{b_1, \ldots, b_k\}$ are always instantiated with the branching nodes in $I$.
The first step of the translation associates an existential variable $b_i$ with every state of the automaton which has more than one child (and thus represents a branching point). Then we enumerate all simple (loopless) paths in the automaton starting from a leaf state, say $q_i$, to a final state, say $q_f$, and record the structural constraints over these linear segments, $\phi_{if}(B, PV, Y)$, which capture the relative positions (over relations $\next$, $\next^+$ and $\next^*$) of the pointer/universal and branching variables with respect to each other and the data formula annotated at the final state $f(q_f)$.
These structural constraints can be constructed as described in~\cite{CAVQDA}. After consdering every pair of such states, the formula corresponding to an \EQSDA\ can be expressed as
\[I =
\exists  b_1 \ldots b_k. \forall y_1 \ldots y_{\ell}. \big(
\bigwedge_f (\wedge_i \phi_{if} \Rightarrow f(q_f)) \wedge (\bigvee_f \wedge_i \phi_{if} )
\big)
\]
A key property in the decidable fragment of \Strand is that universal quantification is not permitted to be over elements that are only a bounded distance away from each other. See~\cite{CAVQDA} for a proof that the structural constraints $\phi_{if}$ are such that $I$ falls in the decidable fragment of \Strand.

\section{Experimental Evaluation}

We implemented the abstract domain over \QSDA s and \EQSDA s presented in this paper, and evaluated them on several list-manipulating programs. We now first present the implementation details followed by our experimental results.
Our prototype implementation along with the experimental results and programs can be found at \url{http://web.engr.illinois.edu/~garg11/qsdas.html}.


\paragraph{\bf Implementation details.}
Given a program $P$ we compute the abstract semantics of the program over the abstract domain $\mathcal{Q_F}^{el}$ over \EQSDA s.
A program is a sequence of statements as defined by the grammar in Figure~\ref{grammar}. In addition to those statements, a program is also annotated with a pre-condition and a bunch of assertions.
The pre-condition formulas belong to a fragment of \Strand over lists and can express quantified properties like sortedness of lists, etc. Given a pre-condition formula $\varphi$, we construct the smallest \EQSDA\ (with respect to the partial-order defined on the \QSDA s) which accepts all the heap skinny-trees which satisfy $\varphi$.
This \EQSDA\ gives us an abstraction of the initial configurations of the program. Starting from these configurations we compute the abstract semantics of the program over $\mathcal{Q_F}^{el}$. The assert statements in the program are ignored during the fix-point computation. Once the convergence of the fix-point has been achieved, the \EQSDA s can be converted back into decidable \Strand formulas (as described in Section~\ref{sec-translation}) and the \Strand decision procedure can be used for validating the assertions.

We recall that the abstract domain $\mathcal{Q_F}^{el}$ is an abstraction of $\mathcal{Q_F}$. So, as much as possible, we want to compute the abstract semantics over the more concrete domain out of the two, i.e. $\mathcal{Q_F}$.
Therefore, for every statement in the program we apply the abstract transformer $F^\sharp$  (and not the more abstract $F_{el}^\sharp$). The intermediate semantic facts (\QSDA s) in our analysis are thus not necessarily elastic.
However to ensure convergence of the analysis, the \QSDA s at the header of the loops are first abstracted to elastic \QSDA s using $\alpha^{el}$ before the join.


Our abstract domains are parameterized by a quantifier-free domain $\F$ over the data formulas. In our experiments, we instantiate $\F$ with the octagon abstract domain~\cite{octagon} from the Apron library~\cite{apron}. It is sufficient to capture the pre/post-conditions and the invariants of all our programs.

\begin{table*}[thb]
	\centering

{\scriptsize
\begin{tabular}{|l| r | r | r | r|| r | r |r||}
	\hline
	~~~Programs & \#PV & \#Y & \#DV & Property~~ & \#Iter & Max. size & Time (s) \\
	 & & & & checked~~~ & & of QSDA & \\  \hline \hline

{\sc init}	& 2 & 1 & 1 & {\sc Init, List} & 4 & 19 & 0.0 \\	\hline
{\sc add-head} & 2 & 1 & 1 & {\sc Init, List} & - & 11 & 0.1  \\ \hline
{\sc add-tail}  & 3 & 1 & 1 & {\sc Init, List} & 4 & 29 & 0.1  \\ \hline
{\sc delete-head} & 2 & 1 & 1 & {\sc Init, List} & - & 10 & 0.0 \\ \hline
{\sc delete-tail}  & 4 & 1 & 1 & {\sc Init, List} & 5 & 51 & 0.5  \\ \hline
{\sc max}	& 2 & 1 & 1 & {\sc Max, List} & 4 & 19 & 0.1   \\ \hline
{\sc clone}    & 4 & 1 & 1 & {\sc Init, List} & 4 & 44 & 0.7  \\ \hline
{\sc fold-clone}  & 5 & 1 & 1 & {\sc Init, List} & 5 & 57 &  3.2 \\ \hline
{\sc copy-Ge5}  &  4 & 1 & 0 & {\sc Gek, List} & 9 & 53 & 2.6  \\ \hline
{\sc fold-split} & 3 &  1 & 1 & {\sc Gek, List} & 4 & 33 & 0.3 \\ \hline
{\sc concat}  & 4 & 1 & 1 & {\sc Init, List} & 5 & 44 & 0.7 \\ \hline
{\sc sorted-find} & 2 & 2 & 2 & {\sc Sort, List} & 5 & 38 & 0.3 \\ \hline
{\sc sorted-insert} & 4 & 2 & 1 & {\sc Sort, List} & 6 & 163 & 5.8 \\ \hline
{\sc bubble-sort} & 4 & 2 & 1 & {\sc Sort, List} & 5/18 & 191 & 42.8 \\ \hline
{\sc sorted-reverse} & 3 & 2 & 0 & {\sc Sort, List} & 5 & 43 & 1.5 \\ \hline
{\sc expressOS-lookup-prev} & 3 & 2 & 1 & {\sc Sort, List} & 6 & 73 & 2.2 \\ \hline \hline

{\sc gslist-append} & 4 & 0 & 1 & {\sc List} & 8 & 3 & 0.0 \\ \hline
{\sc gslist-prepend} & 2 & 0 & 1 & {\sc List} & - & 3 & 0.0  \\ \hline
{\sc gslist-last} & 3 & 0 & 0 & {\sc Last, List} & 3 & 7 & 0.0 \\ \hline
{\sc gslist-free} & 3 & 0 & 0 & {\sc Empty, List} & 1 & 3 & 0.0 \\ \hline
{\sc gslist-position} &4 & 0 & 0 & {\sc List} & 3 & 13 &  0.0 \\ \hline
{\sc gslist-reverse} & 3 & 0 & 0 & {\sc List} & 3 & 5 & 0.0 \\ \hline
{\sc gslist-custom-find} &3 & 1 & 1 & {\sc Gek, List}  & 4 & 29 & 0.1 \\ \hline
{\sc gslist-nth} & 3 & 0 & 1 & {\sc List} & 3 & 7 & 0.0 \\ \hline
{\sc gslist-remove} & 4 & 0 & 1 & {\sc List} & 4 & 10 & 0.0 \\ \hline
{\sc gslist-remove-link} & 5 & 0 & 0 & {\sc List} & 4 & 16 & 0.0 \\ \hline
{\sc gslist-remove-all} & 5 & 1 & 1 & {\sc Gek, List} & 5 & 51 & 0.6 \\ \hline
{\sc gslist-insert-sorted} & 5 & 2 & 1 & {\sc Sort, List} & 6 & 279 & 27.4 \\ \hline

	\end{tabular}
}
	\caption{\small Experimental results. Property checked --- {\sc List}: the return pointer points to a list; {\sc Init}: the list is properly initialized with some key; {\sc Max}: returned value is the maximum of all data values in the list; {\sc Gek}: the list (or some parts of the list) have data values greater than or equal to a key $k$; {\sc Sort}: the list is sorted; {\sc Last}: returned pointer is the last element of the list; {\sc Empty}: the returned list is empty.}
	\label{results}
\end{table*}


\paragraph{\bf Experimental results} We evaluate our abstract domain on a suite of list-manipulating programs (see Table~\ref{results}). For every program we report the number of pointer variables (PV), the number of universal variables (Y), the number of data variables (DV) and the property being checked for the program. We also report the number of iterations required for the fixed-point to converge, the maximum size of the intermediate QSDAs and finally the time taken, in seconds, to analyze the programs.

The names of the programs in Table~\ref{results} are descriptive, and we only describe some of them. The program {\sc copy-Ge5} is from~\cite{celia} and copies from a list only those entries into a new list whose data value is greater than or equal to $5$. Similarly, the program {\sc fold-split}~\cite{celia} splits a list into two lists--one which has only those entries whose data values are greater than or equal to a key $k$ and the other list with entries whose data value is less than $k$. The program {\sc expressOS-lookup-prev} is a method from the module cachePage in a verified-for-security platform for mobile applications~\cite{asplos13}. The module cachePage maintains a cache of the recently used disc pages as a priority queue based on a sorted list. This method returns the correct position in the cache at which a disc page could be inserted.
The programs in the second part of the table are various methods adapted from the Glib list library which comes with the GTK+ toolkit and the Gnome desktop environment. The program {\sc gslist-custom-find} finds the first node in the list with a data value greater  or equal to $k$ and {\sc gslist-remove-all} removes all elements from the list whose data value is greater  or equal to $k$. The programs {\sc gslist-insert-sorted} and {\sc sorted-insert} insert a key into a sorted list.

All experiments were completed on an Intel Core i5 CPU at 2.4GHz with 6Gb of RAM. The number of iterations is left blank for programs which do not have loops.
{\sc bubble-sort} program converges on a fix-point after 18 iterations of the inner loop and 5 iterations of the outer loop.
The size of the intermediate QSDAs depends on the number of universal variables and  the number of pointer variables and largely governs the time taken for the analysis of the programs.
For all programs, our prototype implementation computes their abstract semantics in reasonable time.
Moreover we manually verified that the final \EQSDA s in all the programs were sufficient for proving them correct (this validity check for assertions can be mechanized in the future). 
The results show that the abstract domain we propose in this paper is reasonably efficient and powerful enough to prove a large class of programs manipulating singly-linked list structures.

\newpage
\bibliographystyle{abbrv}
\bibliography{sample}
\newpage
\appendix

\section{Formal Semantics of Programs}\label{semantics}

In this appendix we describe the concrete semantics of the primitive statements in our programming language defined in Figure~\ref{grammar}.
%
Let us assume that there is a special program configuration $Error$ to which the progam transitions to on encountering a memory error. 

\begin{definition}[Strongest post-condition $F^\natural$]
Let $\mathcal{C} = (pc, H, pval, dval)$ be a non-$Error$ program configuration where $H = (Loc, \next, \data)$. Then for any statement  $s$, $\mathcal{C'}$ is the strongest post-condition of $\mathcal{C}$ with respect to the statement $s$, written as $\mathcal{C} \xrightarrow{\mathit{s}}_P \mathcal{C'}$, iff $\mathcal{C'} = (pc', H', pval', dval')$ where $H' = (Loc', \next', \data')$ and $pc'$ is the updated program counter and 
one of the following holds:
%

\begin{itemize}

\item $s = p_i := p_j$ (or $s = p_i := nil$) and  $H' = H$ and $dval' = dval$ and $pval' = pval[p_i/pval(p_j)]$ (or $pval' = pval[p_i/pval(nil)])$.


\item $s = p_i := p_j\rightarrow\next$ and $H' = H$ and $dval' = dval$ and if $pval(p_j) = v$ and $(v, u) \in \next$, then $pval' = pval[p_i/u]$. If $v = dirty$ or $v = nil$, then $\mathcal{C'}$ is \emph{Error}.

\item $s = p_i\rightarrow\next := p_j$ (or $s = p_i\rightarrow\next := nil$) and $pval' = pval$ and $dval' = dval$ and $Loc' = Loc$ and $\data' = \data$ and if $pval(p_i) = v$ and $(v, u) \in \next$ and if $pval(p_j) = w$ (or $pval(nil) = w$), then $\next' = \next \backslash \{(v, u)\} \cup \{(v, w)\}$. If $v = nil$ or $v = dirty$, then $\mathcal{C'}$ is $Error$.

\item $s = \text{new } p_i$ and $dval' = dval$ and $\Loc' = \Loc \cup \{v\}, v \notin \Loc$, $\next' = \next \cup \{(v, dirty)\}$ and $\data' = \data$ and $pval' = pval[p_i/v]$.

\item $s = d_i := p_i\rightarrow\data$ and $H' = H$ and $pval' = pval$ and  if $pval(p_i) = v$, then $dval' = dval[d_i/\data(v)]$.
If $v = nil$ or $v = dirty$ then $\mathcal{C'}$ is $Error$.

\item $s = p_i\rightarrow\data := data\_expr$ and $pval' = pval$ and $dval' = dval$ and $Loc' = Loc$ and $\next' = \next$ and  if $pval(p_i) = v$ then $\data' = \data[v/data\_expr]$.
If $v = nil$ or if $v = dirty$ then $\mathcal{C'}$ is $Error$.

\item $s = \texttt{skip}$ and $\mathcal{C'} = \mathcal{C}$.

\item $s = \text{assume } (\psi_{struct}) $ and $\mathcal{C'} = \mathcal{C}$ and $\mathcal{C} \models \psi_{struct}$.

\item $s = \text{assume } (\psi_{data})$ and $\mathcal{C'} = \mathcal{C}$ and $\mathcal{C} \models \psi_{data}$.

\end{itemize}

\end{definition}

\section{Proof of Theorem~\ref{thmelastification}}\label{app-elastification}

Note that $\AEL$ is elastic by definition of $\Delta_1^{\text{el}}$. It
is also clear that $\Lval(\A) \subseteq \Lval(\AEL)$ because for each
run of $\A$ using states $q_0 \cdots q_n$ the run of $\AEL$ on the
same input uses sets $S_0 \cdots S_n$ such that $q_i \in S_i$, and by
definition $f(q_n)$ implies $f^{\text{el}}(S_n)$.

Now let $\B$ be an EQSDA with $\Lval(\A) \subseteq \Lval(\B)$. Let $t$ be a valuation tree 
accepted by $\AEL$ and let $S$ be the
state of $\AEL$ reached on reading $t$.  We want to show that $t \in
\Lval(\B)$.  Let $p$ be the state reached in $\B$ on $t$. We show that
$f(q)$ implies $f_\B(p)$ for each $q \in S$. From this we obtain
$f^{\text{el}}(S) \Rightarrow f_\B(p)$ because $f^{\text{el}}(S)$ is
the least formula that is implied by all the $f(q)$ for $q \in S$.


Pick some state $q \in S$. By definition of $\Delta^{\text{el}}$ we
can construct a valuation tree $t' \in \Lval(\A)$ that leads to the state $q$ in
$\A$ and has the following property: if all letters of the form
$(\blank,d)$ and those that have a single child are removed from $t$ and from $t'$, then the two
remaining trees have the same symbolic trees. In other words, $t$ and $t'$ can be
obtained from each other by inserting and/or removing
$\blank$-letters.

Since $\B$ is elastic, $t'$ also leads to $p$ in $\B$. From this we
can conclude that $f(q) \Rightarrow f(p)$ because otherwise there
would be a model of $f(q)$ that is not a model of $f(p)$ and by
changing the data values in $t'$ accordingly we could produce an input
that is accepted by $\A$ and not by $\B$.
\qed

\section{Construction of the Strengthen Operator}\label{app-strengthen}

Given a QSDA $\A$, we give below a high level sketch of how to construct the QSDA $\A'$ accepting the following language of formula words: $(Strengthen)~ y~ L_f(\A)$.
The construction of the QSDA $\A'$ takes place in two steps. 

In the first step, we construct a QSDA $\A_1$ which accepts ($\Sigma \times (Y \cup \{-\} \backslash \{y\})$)-labelled formula trees of the form $\big\{(t' \downharpoonright_y, \phi) \mid \phi: \bigsqcap \{ \exists d. \varphi[y\rightarrow\data/d] \mid (t, \varphi) \in L_f(\A), t \downharpoonright_y = t' \downharpoonright_y \} \big\}$. 

And in the second step, we take the cross-product of this automaton $\A_1$ with the initial automaton $\A$ to get the QSDA $\A'$ such that if the symbolic tree $t' \downharpoonright_y$ is mapped by $\A_1$ to data-formula $\phi$ and the symbolic tree $t'$ is mapped by $\A$ to the data-formula $\varphi''$, then $t'$ is mapped in the new automata $\A'$ to the data-formula $\phi \sqcap \varphi''$.
The required cross-product of two QSDAs is similar to the algorithm for computing intersection of tree automata.
The automata which computes the cross-product simulates the transitions of both the automata $\A$ and $\A_1$. 
However, since $\A_1$ accepts trees which do not have the variable $y$ labeling them, the cross-product automata on a label $\pi$ which contains $y$ simulates the transitions of $\A$ on the label $\pi$, but simulates the transition of $\A_1$ on the label $\pi \backslash \{y\}$.
 The states in the cross-product automaton are labeled with the meet of the data-formulas labeling the states of the two individual automata. 

So now let us describe the construction of the automaton $\A_1$. The QSDA $\A_1$ accepts symbolic trees of the form $t \downharpoonright_y$ if $(t, \varphi) \in L_f(\A)$. Hence, $\A_1$ is an automaton which simulates the transitions of $\A$ on a symbolic tree, except when it reads a node in the tree labeled with label $\pi$ which contains variable $y$. When this happens, $\A_1$ simulates the transitions of $\A$ on the label $\pi \backslash \{y\}$.
The data-formulas mapped to each state in $\A_1$ is obtained by existentially quantifying out $y\rightarrow\data$ from the data-formulas mapping the corresponding states in $\A$. 
Note that the transition relation of $\A_1$ we just described might be non-deterministic; the same symbolic tree $t \downharpoonright_y$ might be mapped to more than one data-formula eg. $\exists d. \varphi_1[y\rightarrow\data/d]$ and $\exists d. \varphi_2[y\rightarrow\data/d]$ and so on. We want the automata $\A_1$ to map $t \downharpoonright_y$ to the meet of all these formulas. This can be achieved by determinizing the transition relation of the QSDA, very similar to the determinization procedure of a tree automata. In addition, for a set of states in the deterministic automata, we label it with the meet of the data-formulas labeling each state in the set in the original non-deterministic automata.
In this way, the transition relation could be determinized to obtain QSDA $\A_1$. This completes the construction of $\A'$.

\section{Construction of the Abstract Transformer}\label{app-construction}

First let us introduce some preliminary notation.
For a set S, let $S^i$ denote the $i^{th}$ cartesian power of $S$.
For an n-tuple $x = (x_1, \ldots, x_n) \in X = X_1 \times \ldots \times X_n$, let us define $x \downarrow_i$ for $1 \leq i \leq n$ as the projection onto the $i^{th}$ component of the tuple i.e. $(x_1, \ldots, x_n)\downarrow_i = x_i$. 
When the set $X_i$ can be uniquely distinguished from all other components $X_j$ of the $n$-tuple $X$, $x \downarrow_{X_i} = x_i$ is also used to denote the $X_i^{th}$ component of the tuple $x$. For a tuple $x$, $x' = x[X_1/x_1, \ldots, X_k/x_k]$ is used to denote the tuple which is same as $x$ except at the $X_1, \ldots, X_k$ components where the tuple $x'$ takes the value $x_1, \ldots, x_k$ respectively.

For a function $F: D \rightarrow X$, where $X = X_1 \times \ldots \times X_n$, which maps a domain $D$ to an $n$-tuple $X$, let $F\downarrow_i: D \rightarrow X_i$ for $1 \leq i \leq n$ be defined such that $F\downarrow_i(d) = F(d)\downarrow_i$ for all $d \in D$.
Also for a function $f$, let us denote $f \mid_D$ as the function which is same as $f$ but with its domain restricted to the set $D$. for a given function $f$, let us also define $f' = f[d_1/v_1, \ldots, d_k/v_k]$ as the function which is same as $f$ except at the domain elements $d_1$, \ldots, $d_k$ where $f'$ takes the value $v_1$, \ldots, $v_k$.

We now present the construction of the abstract transformer for each case of the statement $stmt$. 
Let the input \QSDA\ be of the form $(Q, \Pi, \Delta, \mathcal{T}, f)$ where $\Pi = (\Sigma \times (Y \cup \{-\}))$ and $\Sigma = 2^{PV}$. Let us view $\Sigma$ as a set of boolean vectors where the $i^{th}$ bit in a vector indicates whether the pointer $p_i$ belongs to the vector or not.
The output \QSDA\ after the execution of $stmt$ is of the form $(Q', \Pi, \Delta', \mathcal{T'}, f')$ where:
\newline\newline
{\bf Case 1} ($s = p_i := p_j$): 
The evaluation function $f'$ of the automaton depends on whether $p_j$ co-occurs with ${nil}$ or not. To facilitate this, we split each state into two states $Q' = Q \cup \{(q, nil)~|~q \in Q)\}$. 
 Regarding the transition relation, for all transitions $\Delta(q_1, .., q_p, \pi) = q$,
\begin{itemize}
\item if $p_j \not\in \mathcal{T}(q)$, $\Delta'(q_1, .., q_p, \pi[p_i/0]) = q$.
\item if $p_j \in \mathcal{T}(q)$ and $\pi \downarrow_{p_j} = 1$ and $\pi \downarrow_{p_{nil}} = 1$, then $\Delta'(q_1, .., q_p,$ $\pi[p_i/1])$ $  = (q, nil))$.
\item if $p_j \in \mathcal{T}(q)$ and $\pi \downarrow_{p_j} = 1$ and $\pi \downarrow_{p_{nil}} = 0$, then $\Delta'(q_1, .., q_p,$ $\pi[p_i/1])$ $ = q$.

\item if $p_j \in \mathcal{T}(q)$ and $\pi \downarrow_{p_j} = 0$ then there exists a state $q_j$, $1 \leq j \leq p$ such that $p_j \in \mathcal{T}(q_j)$. Correspondingly, we add transitions $\Delta'(q_1, ..., q_j, ..., q_p,$ $\pi[p_i/0])$ $ = q$ and $\Delta'(q_1, ..., (q_j, nil), ..., q_p,$ $\pi[p_i/0])$ $ = (q, nil)$.
\end{itemize}

The type $\mathcal{T'}$ for every state $\in Q'$ is:

$\mathcal{T'}(q) = \mathcal{T'}(q, nil)= \begin{cases}
        \mathcal{T}(q) \cup \{p_i\} & \text{if $p_j \in \mathcal{T}(q), p_i \notin \mathcal{T}(q)$ }\\
	  \mathcal{T}(q) \backslash \{p_i\} & \text{if $p_j \notin \mathcal{T}(q), p_i \in \mathcal{T}(q)$ }\\
	  \mathcal{T}(q) & \text{otherwise}\\
      \end{cases}
$\\
The evaluation function depends on whether the pointer $p_j$ is $nil$ or not. 
Hence, $f'(q) = \exists d. f(q)[p_i\rightarrow\data/d] \sqcap (p_i\rightarrow\data = p_j\rightarrow\data)$. Otherwise, $f'(q, nil) = \exists d. f(q)[p_i\rightarrow\data/d]$.
\newline\newline
{\bf Case 2} ($s = p_i := p_m \rightarrow next$): Firstly, $Q' = \{q~|~ q \in Q, p_m \notin \mathcal{T}(q) \} \cup \{(q,*), (q, v)~|~q \in Q, p_m \in \mathcal{T}(q), v \in PV \cup Y \cup \{-\}\}$. The automaton transitions to a state of the form $(q,*)$ on reading the pointer variable $p_m$. This is like guessing the state in which the automaton transitions to, on reading variable $p_m$. After reading $p_m$, the automaton
reads variable $p_i$ and transitions from state $(q,*)$ to a state of the form $(q',v)$ where $v$ is basically a pointer variable or a universal variable which is co-read with $p_i$. The variable $v$ is used later when defining the evaluation functions for states $Q'$. 

More formally, for all transitions $\Delta(q_1, ..., q_p,\pi) = q$,
\begin{itemize}
\item if $p_m \notin \mathcal{T}(q)$ then $\Delta'(q_1, ..., q_p, \pi[p_i/0]) =  q$.
\item if $p_m \in \mathcal{T}(q)$ and $\pi \downarrow_{p_m} = 1$ then $\Delta'(q_1, ..., q_p$ $, \pi[p_i/0]) = (q,*)$.
\item if $p_m \in \mathcal{T}(q)$ and $\pi \downarrow_{p_m} = 0$ then there exists a state $q_m \in \{q_1, ..., q_p\}$ such that $p_m \in \mathcal{T}(q_m)$. 
Accordingly we add the transition, $\Delta'(q_1, ..., (q_m,*), ..., q_p,$ $\pi[p_i/1]) = (q,v)$ if there exists a variable $v \in PV \cup Y$ such that $\pi \downarrow_{v} = 1$ and $\pi \downarrow_{nil} = 0$. Otherwise the transition  $\Delta'(q_1, ..., (q_m,*), ..., q_p, \pi[p_i/1]) = (q,-)$.

This covers the case when state $q_m$ was reached immediately after reading variable $p_m$. For the other case, we add transitions 
$\Delta'(q_1$ $, ...,$ $(q_m,v)$ $, ..., q_p$ $, \pi[p_i/0]) = (q,v)$ for all $v \in PV \cup Y \cup \{-\}$.
\end{itemize}

Note that the final evaluation formula is only associated with states of the form $(q, v)$ or $ (q,-)$. For all $q \in Q$ such that $p_m \in \mathcal{T}(q)$ and $v \in PV \cup Y$, $f'(q,v) = \exists d. f(q)[p_i\rightarrow\data/d] \sqcap (p_i\rightarrow\data = v\rightarrow\data)$. Otherwise, $f'(q,-) = \exists d. f(q)[p_i\rightarrow\data/d]$.
For all other states, the evaluation formula is $false$.

Finally, for all $q' \in Q'$, the type associated with $q'$ is given as:\\
$\mathcal{T'}(q') = \begin{cases}
        \mathcal{T}(q') \backslash \{p_i\} & \text{if $q' \in Q, p_m \notin \mathcal{T}(q')$ }\\
	  \mathcal{T}(q) \backslash \{p_i\} & \text{if $q' = (q,*), q \in Q$ }\\
	  \mathcal{T}(q) \cup \{p_i\}& \text{if $q' = (q, v), q \in Q, v \in PV \cup Y \cup \{-\}$}\\
							
      \end{cases}
$\\
\newline\newline
{\bf Case 3} ($s = new$ $p_i$): The statement $s$ allocates a new node which is pointed to by variable $p_i$ and is added as a child to the root of data trees accepted by the original automaton. The universal variables, apart from the nodes already present in the data trees, now also have to valuate over the newly allocated node. The set of states of the new automaton is $Q' = (Q \cup \{\hat{q}\}) \times \hat{Y}$ where $\hat{q} \notin Q$. The states of the form $(Q \cup \{\hat{q}\}, y)$ are used to accept valuation trees where univeral variable $y$ valuates over the newly allocated node whereas the other universal variables valuate over the existing nodes present in the heap tree. The states $(Q \cup \{\hat{q}\}, -)$ are used to accept valuation trees where none of the universal variable valuates to the newly allocated node. The state $\hat{q}$ is used to transition the automaton to a special state on reading the new node labeled with pointer variable $p_i$ i.e. $\Delta'((\{p_i\}, \hat{y})) = (\hat{q}, \hat{y})$ for all $\hat{y} \in \hat{Y}$. Also for all transitions $\Delta(q_1, ..., q_p, \pi) = q$,
\begin{itemize}
\item if $\pi \downarrow_{{root}} = 0$, $\Delta'((q_1, \hat{y}), ..., (q_p, \hat{y}), \pi[p_i/0, \hat{y}/0]) = (q, \hat{y})$ for all $\hat{y} \in \hat{Y}$. 
\item if $\pi \downarrow_{{root}} = 1$, $\Delta'((q_1, \hat{y}), ..., (q_p, \hat{y}), (\hat{q}, \hat{y}), \pi[p_i/0, \hat{y}/0]) =  (q, \hat{y})$ for all $\hat{y} \in \hat{Y}$.
\end{itemize}

The final evaluation formula is given as: \\
$f'(q, \hat{y}) = \begin{cases}
        \exists d_1 d_2. f(q)[p_i\rightarrow\data/d_1, y\rightarrow\data/d_2] \sqcap (p_i\rightarrow\data = y\rightarrow\data)\\
        ~~~~~~~~~~~~~~~~~~~~~~~~~~~~~~~~~~~~~~~~~~~~~ \text{if $q \in Q, \hat{y} = y \in Y$}\\
        \exists d. f(q)[p_i\rightarrow\data/d]   ~~~~~~~~~~~~~~~~~ \text{if $q \in Q, \hat{y} = -$}\\
        \text{\emph{false}}    ~~~~~~~~~~~~~~~~~~~~~~~~~~~~~~~~~~~~~~~ \text{if $q = \hat{q}$}\\
      \end{cases}
$\\

Also the types for each state in the new automaton are:\\
$\mathcal{T'}(q, \hat{y}) = \begin{cases}
	  \{p_i\} & \text{if $q = \hat{q}, \hat{y} = -$}\\
	  \{p_i, y\} & \text{if $q = \hat{q}, \hat{y} = y \in Y$}\\
        \mathcal{T}(q) \backslash \{p_i\} & \text{if $q \in Q, {root} \notin \mathcal{T}(q), \hat{y} = -$ }\\
        \mathcal{T}(q) \backslash \{p_i, y\} & \text{if $q \in Q, {root} \notin \mathcal{T}(q), \hat{y} = y \in Y$ }\\
	  \mathcal{T}(q)  & \text{if $q \in Q, {root} \in \mathcal{T}(q)$ }\\
	  
      \end{cases}
$\\
\newline\newline
{\bf Case 4} ($s = p_m \rightarrow next := p_i$):
Firstly, $Q' = Q \cup (Q,*) \cup \{(q_1, q_2)~|~q_1, q_2 \in Q, p_i \in \mathcal{T}(q_1) \text{\emph{ iff }} p_m \notin \mathcal{T}(q_1)\}$. 
From the semantics of the strongest-post, we know that the new automaton removes any sub-tree rooted at $p_m$ and attaches it as an additional child to a node labelled with variable $p_i$.
States of the form $(Q,*)$ are special states in which the automaton transitions to on reading the variable $p_m$.
If $(q_2, *)$ accepts a tree $\tau_m$ rooted at $p_m$ then the state $(q_1, q_2)$ where $p_i \in \mathcal{T}(q_1)$ accepts a tree which has $\tau_m$ as an additional child to an internal node labelled with $p_i$. On the other hand, if $p_m \in \mathcal{T}(q_1)$ then $(q_1, q_2)$ accepts a tree which had its subtree $\tau_m$, rooted at $p_m$, removed.
Describing the transition relation in detail, for a transition $\Delta(q_1, ..., q_p, \pi) = q$:
\begin{enumerate}
\item if $\{p_i, p_m\} \cap \mathcal{T}(q) = \phi$,  then we add the same transition to the new automaton i.e. $\Delta'(q_1, ..., q_p, \pi) = q$.
\item if $p_m \in \mathcal{T}(q)$ and $p_i \notin \mathcal{T}(q)$
    \begin{itemize}
	\item and if $\pi \downarrow_{p_m} = 1$ then the automaton should transition to a state of the form $(Q,*)$; therefore $\Delta'(q_1, ..., q_p, \pi) = (q,*)$. 
  	\item otherwise, there exists a state $q_m \in \{q_1, ..., q_p\}$ such that $p_m \in \mathcal{T}(q_m)$. In case $q_m$, in the original automaton, accepted trees rooted at $p_m$, the new automaton should remove $q_m$ from the left hand side of the transition and should transition to a state $(q, q_m)$ via $\Delta'(q_1,$ $..,$ $q_{m-1},$ $q_{m+1},$ $...,$ $q_p, \pi) = (q,q_m)$. To handle the other case, where $q_m$ accepted trees which were not rooted at $p_m$, the transitions $\Delta'(q_1,$ $...,$ $(q_m,\hat{q}),$ $...,$ $q_p,$ $\pi) = (q,\hat{q})$ are added for all $\hat{q} \in Q$.
    \end{itemize}

\item if $p_i \in \mathcal{T}(q)$ and $p_m \notin \mathcal{T}(q)$
	\begin{itemize}
	\item and if $\pi \downarrow_{p_i} = 1$ then the new automaton should accept a tree at $q$ which has an additional child $\tau_m$ rooted at $p_m$. Since all trees rooted at $p_m$ are accepted at states of the form $(Q,*)$ (the first subcase of 2 above), transition $\Delta'(q_1, ..., q_p, (q_m,*), \pi) = (q, q_m)$ is added for all $q_m \in Q$.
	\item otherwise, there exists a state $q_i \in \{q_1, ..., q_p\}$ such that $p_i \in \mathcal{T}(q_i)$ and the fact, that any node labelled with $p_i$ in the tree accepted at $q$ has as an additional child a tree $\tau_m$ rooted at $p_m$, is propagated recursively via the transitions $\Delta'(q_1, ..., (q_i, q_m), ..., q_p, \pi) = (q, q_m)$ for all $q_m \in Q$. 
	\end{itemize}

\item if $\{p_i, p_m\} \subseteq \mathcal{T}(q)$
	\begin{itemize}
	\item and $\pi \downarrow_{p_m} = 1$ (reagrdless of the value of $\pi \downarrow_{p_i}$) no transition is added to $\Delta'$, as any heap configuration accepted by the original automaton via this transition leads to a cycle on the execution of statement $stmt$.
	\item otherwise if $\pi \downarrow_{p_m} = 0$ and $\pi \downarrow_{p_i} = 1$, there will exist a state $q_m \in \{q_1, ..., q_p\}$ such that $p_m \in \mathcal{T}(q_m)$. The corresponding state in the new transition will be $(q_m, \hat{q})$ if $\hat{q}$ was the state of the original automaton which accepted the internal subtree $\tau_m$ rooted at $p_m$. Since $\pi \downarrow_{p_i} = 1$, an additional state $(\hat{q}, *)$ is added to the left hand side of the transition to ensure that the new automaton accepts the tree which has $\tau_m$ as an additional child to a node labelled with $p_i$. Formally, $\Delta'(q_1, ..., (q_m, \hat{q}),$ $...,$ $q_p, (\hat{q},*), \pi) = q$ for all $\hat{q} \in Q$.

	\item otherwise if $\pi \downarrow_{p_m} = \pi \downarrow_{p_i} = 0$ and there exists a state $q_{im} \in \{q_1, ..., q_p\}$ such that $\{p_i, p_m\} \subseteq \mathcal{T}(q_{im})$ then the transition remains unchanged i.e. $\Delta'(q_1,$ $...,$ $q_{im},$ $...,$ $q_p, \pi) =  q$.

	\item otherwise if $\pi \downarrow_{p_m} = \pi \downarrow_{p_i} = 0$ and there exist states $q_i, q_m \in \{q_1, ..., q_p\}$ such that $p_i \in \mathcal{T}(q_i)$ and $p_m \in \mathcal{T}(q_m)$, then the transition   $\Delta'(q_1,$ $...,$ $(q_i,\hat{q}),$ $...,$ $(q_m,\hat{q}),$ $...,$ $q_p, \pi) = q$ is added for all $\hat{q} \in Q$. Note that $(q_i, \hat{q})$ accepts a tree which has an additional child (accepted at $(\hat{q},*)$) at a node labelled with $p_i$ (explained in case 3 above) and $(q_m, \hat{q})$ accepts a tree where the internal subtree rooted at $p_m$ and accepted at $(\hat{q},*)$ has been removed (explained in the second subcase of 2 above). Note that if $q_m = \hat{q}$ then the state $(q_m, \hat{q})$ is removed from the left hand side of the transition $\Delta'$ i.e. $\Delta'(q_1,$ $...,$ $(q_i,\hat{q}),$ $...,$ $q_p, \pi) = q$.
    	\end{itemize}

\end{enumerate}

The final evaluation formula is unchanged for the states $Q \subseteq Q'$, it is \emph{false} for the newly added states i.e. $f'(q) = f(q), q \in Q$ and $f'(q,*) = f'(q_1, q_2) =$ \emph{false}. The type $\mathcal{T'}$ for the new automaton is defined as:\\
$\mathcal{T'}(q) = \begin{cases}
        \mathcal{T}(q) & \text{if $q \in Q$ }\\
	  \mathcal{T}(\hat{q})  & \text{if $q = (\hat{q},*), \hat{q} \in Q$ }\\
	  \mathcal{T}(q_1) \cup \mathcal{T}(q_2) & \text{if $q = (q_1, q_2),  p_i \in \mathcal{T}(q_1)$}\\
	  \mathcal{T}(q_1) \backslash \mathcal{T}(q_2) & \text{if $q = (q_1, q_2), p_m \in \mathcal{T}(q_1)$}\\
      \end{cases}
$\\
\newline\newline
{\bf Case 5} ($s = p_m \rightarrow\data := a$): On execution of this statement, the structure component of the data trees accepted by the automaton is unchanged; however the final evaluation function has to now record the fact that the value of the data pointed by variable $p_m$ is assigned  the value of $a$. If, for a particular valuation tree, $p_m$ is co-read with variable $v \in PV \cup Y$ before it is accepted at state $q$, $f'(q)$ should also record that  the data value pointed by variable $v$  is now assigned to $a$. So the new automaton needs to track the set of variables which are co-read with $p_m$ for a particular valuation tree. Hence $Q' = \{q~|~q \in Q, p_m \notin \mathcal{T}(q)\} \cup \{(q, S)~|~q \in Q, p_m \in \mathcal{T}(q), S \subseteq PV \cup Y\}$. Regarding the transition relation, for all transitions $\Delta(q_1, ..., q_p, \pi) = q$,
\begin{itemize}
\item if $p_m \notin \mathcal{T}(q)$, $\Delta'(q_1, ..., q_p, \pi) =  q$.
\item otherwise if $p_m \in \mathcal{T}(q)$ 
	\begin{itemize}
	\item and if $\pi \downarrow_{p_m} = 1$ then $\Delta'(q_1, ..., q_p, \pi) =  (q, S)$ where $\forall s \in S.~ \pi \downarrow_s = 1$. 
	\item otherwise if $\pi \downarrow_{p_m} = 0$ then there exists a state $q_m \in \{q_1, ..., q_p\}$ such that $p_m \in \mathcal{T}(q_m)$. Consequently, we add transitions $\Delta'(q_1, ..., (q_m, S), ... , q_p, \pi) = (q, S)$ for all $S \subseteq PV \cup Y$.
	\end{itemize}
\end{itemize}

The final evaluation function $f'$ is given as: $f'(q) = f(q)$ for all $q \in Q \cap Q'$; $f'(q, S) = \exists d. f(q)[v_1\rightarrow\data/d, \ldots, v_\ell\rightarrow\data/d] \sqcap (v_1\rightarrow\data = a) \sqcap \ldots \sqcap (v_\ell\rightarrow\data = a)$ where $S = \{v_1, ..., v_\ell\}$ and includes variable $p_m$. The type for each state of the new automaton is also same as the type in the original automaton i.e. $\mathcal{T'}(q) = \mathcal{T}(q)$ for all $q \in Q \cap Q'$; while $\mathcal{T}(q, S) = \mathcal{T}(q)$ for all $S$ in the remaining states.
\newline\newline
{\bf Case 6} ($s =$ \emph{assume} $(p_i = p_j)$ ): 
The output QSDA is obtained by removing from the input QSDA, transitions where variables $p_i$ and $p_j$ do not occur together. Formally, $Q' = Q, f' = f, \mathcal{T'} = \mathcal{T}$ and for all transitions $\Delta(q_1, ..., q_p, \pi) = q$ in the input QSDA, $\Delta'(q_1, ..., q_p, \pi) = q$ iff $\pi \downarrow_{p_i} = \pi \downarrow_{p_j}$. 
\newline\newline
{\bf Case 7} ($s =$ \emph{assume} $(p_i \neq p_j)$ ): 
The output QSDA is obtained by removing from the input QSDA, transitions where variables $p_i$ and $p_j$  occur together. Formally, $Q' = Q, f' = f, \mathcal{T'} = \mathcal{T}$ and for all transitions $\Delta'(q_1, ..., q_p, \pi) =  q$ in the input QSDA, $\Delta'(q_1, ..., q_p, \pi) = q$ iff $\pi \downarrow_{p_i} = 0$ or $\pi \downarrow_{p_j} = 0$ or both. 
\newline\newline
{\bf Case 8} ($s = $ \emph{assume} $\psi_{data}$ ): 
In this case, $Q' = Q$ and the transitions $\Delta'$ is same as $\Delta$.
The type 
$\mathcal{T'} = \mathcal{T}$ and for all $q \in Q, f'(q) = f(q) \sqcap \psi_{data}$.
\newline\newline
\noindent The transition relation $\Delta'$ thus constructed might need to be determinized to obtain $\mathcal{A'}$. 
For a symbolic tree $t$ and formulas $\varphi_1, \ldots, \varphi_j$ such that $(t, \varphi_1), \ldots, (t, \varphi_j)$ belong to the language, the determinization procedure maps $t$ to the formula $\varphi_1 \sqcup \ldots \sqcup \varphi_j$. The determinization procedure is similar to the powerset construction for determinizing tree automata; except for any set of states, the final evaluation function is now assigned to be the join of the formulas being mapped to the individual states in the set.

\end{document}